\newif\ifarxiv
\arxivtrue

\documentclass[a4paper,USenglish,11pt]{article}
\usepackage{amsmath,amsfonts,amssymb}
\DeclareMathOperator*{\argmax}{arg\,max}
\DeclareMathOperator*{\argmin}{arg\,min}
\usepackage{pifont}
\usepackage{amsthm}
\usepackage{amssymb}
\usepackage[pagebackref]{hyperref}
\usepackage{graphicx}
\usepackage{caption}
\usepackage{subcaption}
\usepackage{bm}
\usepackage{bbm}
\usepackage{diagbox}
\usepackage[inline]{enumitem}
\usepackage[misc]{ifsym}
\usepackage{geometry}
\usepackage{cleveref}

\usepackage{placeins}

\usepackage{xspace}
\usepackage[textwidth=20mm,obeyFinal,colorinlistoftodos]{todonotes}

\usepackage{marginnote}

\newcommand{\arxiv}[1]{%
	\ifarxiv{}%
	#1%
	\else{}%
	\fi{}%
}

\newcommand{\mytodo}[2]{\todo[size=\tiny, color=#1!50!white]{#2}\xspace}
\newcommand{\myrevtodo}[2]{{%
 \let\marginpar\marginnote
 \reversemarginpar
 \renewcommand{\baselinestretch}{0.8}%
 \todo[size=\tiny, color=#1!50!white]{#2}}}
\newcommand{\myinlinetodo}[2]{\todo[size=\small, color=#1!50!white, inline, 
caption={}]{#2}\xspace}
\newcommand{\registerAuthor}[3]{%
  \expandafter\newcommand\csname #2com\endcsname[1]{\mytodo{#3}{\textsc{#2}: 
  ##1}}%
  \expandafter\newcommand\csname 
  #2revcom\endcsname[1]{\myrevtodo{#3}{\textsc{#2}: ##1}}%
  \expandafter\newcommand\csname 
  #2inline\endcsname[1]{\myinlinetodo{#3}{\textsc{#2}: ##1}}%
  \expandafter\newcommand\csname 
  #2inlineLater\endcsname[1]{\lv{\myinlinetodo{#3}{\textsc{#2}: ##1}}}%
}

\registerAuthor{Piotr Faliszewski}{pf}{orange}
\registerAuthor{Niclas Boehmer}{nb}{green}
\registerAuthor{Robert bredereck}{rb}{yellow}

\makeatletter
\providecommand*{\cupdot}{%
	\mathbin{%
		\mathpalette\@cupdot{}%
	}%
}
\newcommand*{\@cupdot}[2]{%
	\ooalign{%
		$\m@th#1\cup$\cr
		\hidewidth$\m@th#1\cdot$\hidewidth
	}%
}

\providecommand{\keywords}[1]
{
	{\small
	\textbf{\textit{Keywords---}} #1}
}

\newtheorem{axiom}{Axiom}
\newtheorem{theorem}{Theorem}

\newtheorem{proposition}{Proposition}

\AtBeginDocument{}
\AtBeginDocument{}
\AtBeginDocument{}
\AtBeginDocument{}
\AtBeginDocument{}

\DeclareMathOperator{\ut}{ut}
\DeclareMathOperator{\eg}{eg}
\DeclareMathOperator{\har}{har}
\DeclareMathOperator{\Ihar}{ihar}

\DeclareMathOperator{\seq}{seq}
\DeclareMathOperator{\fix}{fix}

\DeclareMathOperator{\mi}{min}
\DeclareMathOperator{\ma}{max}
\DeclareMathOperator{\score}{score}

\usetikzlibrary{arrows.meta, calc}
\usepackage{tikzpeople}

\definecolor{field}{RGB}{0,156,0}

\tikzset{
	every node/.style={draw=white, very thick, inner sep=0, outer sep=0},
	every path/.style={draw=white, very thick},
}

\newcommand\area[2]{%
	\begin{scope}[shift={(#1)}, transform shape, rotate=#2]
		\node[minimum width=.55cm,minimum height=1.832cm, anchor=west] 
		(small#2) at 
		(0,0) {};
		\node[minimum width=1.65cm,minimum height=4.032cm, anchor=west] (big#2) 
		at 
		(0,0) {};
		\node[minimum width=.244cm, minimum height=.732cm, anchor=east] 
		(goal#2) at 
		(0,0) {};
		\node[inner sep=.3mm, circle, fill=white] (penalty#2) at (1.1,0) {};
		\begin{scope}
			\tikzset{every path/.style={}}
			\clip (big#2.south east) rectangle ++ (1,5); 
			\draw[white, very thick] (penalty#2) circle (0.915cm);
		\end{scope}
	\end{scope}
}

\begin{document}
\title{\Large \bf Line-Up Elections: \\Parallel Voting with Shared Candidate 
Pool}

\author{Niclas Boehmer$^1$\thanks{Supported by the DFG project MaMu  (NI 
		369/19).} , Robert Bredereck$^1$, Piotr
	Faliszewski$^2$\thanks{Supported  by  a  Friedrich Wilhelm Bessel 
		Award from 
		the Alexander von Humboldt Foundation.} , \\ 
		Andrzej~Kaczmarczyk$^1$\thanks{Supported by 
			the DFG project AFFA (BR 5207/1 and NI 369/15).} , and Rolf 
			Niedermeier$^1$}

\date{
	\small
	$^1$~TU Berlin, Algorithmics and 
	Computational Complexity
	\texttt{\{niclas.boehmer,robert.bredereck,a.kaczmarczyk, 
		rolf.niedermeier\}@tu-berlin.de}\\
	$^2$~AGH University, Krakow\\
	\texttt{faliszew@agh.edu.pl}\\
}

	\maketitle 
	\begin{abstract}
	 We introduce the model of line-up elections which captures parallel or 
	 sequential single-winner elections with a shared candidate pool. The goal 
	 of a line-up election is to find a high-quality assignment of a 
	 set of candidates 
	 to a set of positions such that each position is filled by exactly one 
	 candidate and each candidate fills at most one position. A score for each 
	 candidate-position pair is given as part of the input, which expresses the 
	 qualification of the candidate to fill the position. 
	 We propose
	 several voting rules for line-up elections and analyze them from an 
	 axiomatic 
	 and an
	 empirical perspective using real-world data from the popular video game 
	 FIFA.
	\end{abstract}
\hspace{17pt} \keywords{Single-winner voting, Multi-winner voting,
Assignment
	problem, \\ \hspace*{20pt} Axiomatic analysis, and Empirical analysis.}
%
%
%
	\section{Introduction}  
	Before the start of the soccer World Cup 2014, Germany's head coach
	Joachim L{\"o}w had problems to find an optimal team 
	formation. Due to several
	injuries, L{\"o}w was stuck without a traditional striker. He 
	decided to play with three
	offensive midfielders instead, namely, M{\"u}ller, {\"O}zil, and 
	G{\"o}tze.
	However, he struggled to decide who of the
	players should play in the center, on the right, and on the 
	left.\footnote{This story and the opinions of the coaches are 
		fictional. 
		However, L{\"o}w really faced the described problem before the World 
		Cup 2014.} At the final
	coaching meeting, he surveyed the opinions of ten 
	coaching assistants asking for each of the candidates, ``Is this candidate 
	suitable to play on the left/in the center/on
	the right?''. Coaching assistants were allowed to approve an 
	arbitrary subset of candidate-position pairs. He got answers resulting in 
	the following numbers of approvals
	for each candidate-position pair: 
	\begin{center}
		\begin{tabular}{c|c|c|c|} 
			
			\diagbox{Candidate}{Position}& Left & Center & Right \\ \hline 
			M{\"u}ller & 5
			& 10 & 9\\ {\"O}zil & 3 & 8 & 5 \\ G{\"o}tze &4 & 7 & 4
		\end{tabular}
	\end{center}

	After collecting the results, 
	some 
	of the coaches 
	argued that M{\"u}ller must play in the center, as everyone agreed with 
	this. Others argued that M{\"u}ller should play on the right, as
	otherwise this position would be filled by a considerably less 
	suitable
	player. Finally, someone pointed out that M{\"u}ller should play on the 
	left,
	as this was the only possibility to fill the positions such that every 
	position
	gets assigned a player approved by at least half of the 
	coaches.

	The problem of assigning M{\"u}ller, {\"O}zil, and G{\"o}tze can be modeled 
	as three parallel single-winner elections 
	with 
	a shared candidate pool, where every candidate can win at most one election 
	and
	each voter is allowed to cast different preferences for each 
	election. In our example, the coaches are the voters, the 
	players are the candidates and the three locations on the field are 
	the positions.
	Classical single-winner voting rules do not 
	suffice to determine the winners in such settings, as a candidate may win 
	multiple elections. Also multi-winner voting rules cannot be used, 
	as a voter can asses the candidates differently in different elections. 
	Other examples of parallel single-winner elections with a shared 
	candidate pool include a company that wants to fill different positions 
	after an open call for applications,
	a cooperation electing an executive board composed of different positions,
	and a professor who assigns students to projects. 
	
	In this paper, we introduce a framework for such settings: 
	In a \emph{line-up election}, we get as input 
	a 
	set of candidates, a 
	set of positions, and for each candidate-position pair a score expressing 
	how suitable this candidate is to win the election for this 
	position. The goal of a line-up election is to find a ``good'' assignment 
	of 
	candidates to 
	positions such that each position gets assigned exactly one candidate and 
	each candidate is assigned 
	at most once. There exist multiple possible sources of the scores. For 
	instance, a variety of single-winner voting rules aggregate
	preference profiles into single scores for each candidate and then select 
	the candidate with the highest score as the winner of the given election. 
	Examples of such rules include Copeland's 
	voting rule, where the score of a candidate is the number of her pairwise 
	victories minus the number of her pairwise defeats, positional scoring 
	rules, or  
	Dodgson's voting rule.    
	Thus, line-up elections offer a flexible framework that 
	can be built upon a variety of single-winner voting rules.

	\paragraph{Our Contributions.} This paper introduces \emph{line-up 
	elections}---parallel single-winner elections with a
	shared candidate pool---and initiates a study thereof. After stating the
	problem formally, we propose two 
	classes of voting rules, sequential and OWA-rules. Sequential 
	rules fill positions in some order---which may depend on 
	the scores---and select the best still available candidate for a given 
	position. In 
	the versatile class of OWA-rules, a rule aims at
	maximizing some ordered weighted average (OWA) of the scores of the
	assigned candidate-position pairs. 
	We highlight seven rules from these two classes. 
	Subsequently, inspired by work on 
	voting, we describe several desirable axioms for 
	line-up
	voting rules and provide a comprehensive and diverse 
	picture of their 
	axiomatic properties. We complement this axiomatic 
	analysis by empirical investigations on data from the popular soccer video 
	game FIFA~\cite{fifa} and synthetic data.	
	
	\medskip
	As our model considers 
	multiple, parallel
	single-winner elections, it can be seen as an extension of single-winner 
	elections; indeed, we can view the scores of candidates for a position as 
	obtained from some voting rule 
	\cite{arrow2010handbook,DBLP:reference/scw/BramsF2002,DBLP:reference/choice/Zwicker16}.
	It
	reduces to multi-winner voting~\cite{DBLP:reference/trends/FaliszewskiS17} 
	if every voter casts the same 
	vote
	in all elections. Most of our proposed axioms are 
	generalizations
	of axioms studied in those settings
	\cite{DBLP:journals/scw/AzizBCEFW17,DBLP:journals/scw/ElkindFSS17,DBLP:reference/choice/Zwicker16}.
	Previously, committee elections where the committee consists of different 
	positions were 
	rarely considered. Aziz and 
	Lee~\cite{DBLP:conf/aies/0001L18} studied multi-winner elections where a 
	given 
	committee is
	partitioned into different sub-committees and each candidate is 
	suitable
	to be part of some of these sub-committees.
	
	We defer several details of many discussions and proofs
		to the appendix.
	
	\section{Our Model}
	In a \emph{line-up election} $E$, we are given a set of $m$
	candidates $C=\{c_1,\dots, c_m\}$ and $q$ positions~$P=\{p_1,\dots,p_q\}$ 
	with $m\geq q$, together with a score matrix~$\mathbf{S}\in 
	\mathbb{Q}^{m\times q}$. For~$i\in [m]$ and~$j\in[q]$, $\mathbf{S}_{i,j}$ 
	is the score of candidate $c_i$ for position $p_j$, which we denote 
	as~$\score_{p_j}(c_i)$. An 
	outcome of $E$ is an assignment of candidates to positions, where each 
	position is assigned 
	exactly one candidate and each candidate gets assigned to at most one 
	position. We call an outcome a \emph{line-up} $\pi$ and write, for a 
	position $p\in P$, $\pi_p\in C$ to denote the 
	candidate 
	that is assigned
	to position $p$ in $\pi$. We write a line-up $\pi$ as a 
	$q$-tuple~$\pi=(\pi_{p_1},\dots \pi_{p_q})\in C^q$ with pairwise
    different entries. 
    
    For a 
    candidate-position pair 
    $(c,p)\in
    C\times P$, we say that $(c,p)$ is assigned in $\pi$ if $\pi_p=c$ and write
    $(c,p)\in\pi$. Moreover, for an outcome $\pi$ and a candidate $c$, let
    $\pi(c)\in P\cup \{\square\}$ be the position that $c$ is assigned to in
    $\pi$; that is, $\pi(c)=p$ if $\pi_p=c$ and $\pi(c)=\square$ if~$c$ does 
    not
    occur in $\pi$. We write $c\in \pi$ if $\pi(c)\neq \square$, and $c\notin
    \pi$ otherwise. For a line-up~$\pi$ and a subset of positions $P'\subseteq
    P$, we write $\pi|_{P'}$ to denote the tuple $\pi$ restricted to positions 
    $P'$ and $\pi_{P'}$ to denote the
    set of candidates assigned to positions $P'$ in $\pi$. For a position $p\in 
    P$ and line-up $\pi$, we write $\score_p(\pi)$ 
    to 
	denote $\score_p(\pi_p)$.
	Moreover, we refer to
	$(\score_{p_1}(\pi),\dots, \score_{p_q}(\pi))$ as the \emph{score vector} of
	$\pi$. A \emph{line-up voting rule} $f$ maps a line-up 
	election 
	$E$ to a set
	of winning line-ups, where we use $f(E)\subseteq 2^{C^q}$ to denote the
	set of winning line-ups returned by rule $f$ applied to election
	$E$.

	\section{Line-Up Elections as Assignment Problems} 
	
	It is possible to interpret line-up elections as instances of the 
	Assignment Problem, which aims to find a (maximum weight) matching in a 
	bipartite graph. 
	The assignment graph~$G$ of a line-up election $(\mathbf{S},C,P)$ is a 
	complete weighted bipartite graph $G=(C\uplus P,E,w)$ with edge set $E:=\{ 
	\{c,p\}\mid c\in C \wedge p\in P\}$ 
	and weight function
	$w(c,p):=\score_p(c)$ for~$\{c,p\}\in E$. Every matching in the 
	assignment graph which matches all positions induces a valid line-up.  
	
	The Assignment Problem and its generalizations have been 
	mostly studied from an algorithmic and fairness
	perspective
	\cite{DBLP:journals/ior/Garfinkel71,DBLP:journals/algorithmica/GargKKMM10,DBLP:conf/atal/GoldenP10,DBLP:books/daglib/p/Kuhn10,DBLP:journals/algorithmica/LescaMP19,DBLP:conf/aaai/LianMNW18}.
	For instance, Lesca et al. 
	\cite{DBLP:journals/algorithmica/LescaMP19} studied finding 
	assignments with balanced satisfaction from an 
	algorithmic perspective.
	They utilized ordered weighted average 
	operators and proved that finding assignments that maximize an arbitrary 
	non-decreasing ordered weighted average is NP-hard (see next section for 
	definitions). One generalization of the Assignment 
	Problem 
	is the Conference Paper Assignment Problem (\textsc{CPAP}), which tries to find a 
	many-to-many assignment of papers to reviewers with capacity constraints on 
	both sides \cite{DBLP:conf/aaai/Goldsmith07}. Focusing on egalitarian 
	considerations, Garg et 
	al.~\cite{DBLP:journals/algorithmica/GargKKMM10} studied 
	finding outcomes of \textsc{CPAP} which are 
	optimal 
	for the reviewer that is worst off, where they break ties by looking at the 
	next worst reviewer. They proved that this task is computationally hard in 
	this generalized setting. In contrast to our work and the work by Lesca et 
	al. 
	\cite{DBLP:journals/algorithmica/LescaMP19}, Lian et al. 
	\cite{DBLP:conf/aaai/LianMNW18} employed OWA-operators in the context of 
	\textsc{CPAP} on a different level. 
	Focusing on the satisfaction of reviewers, 
	they studied finding 
	assignments maximizing
	the ordered weighted average 
	of the values a reviewer gave to her assigned papers 
	and conducted experiments 
	where they compare different OWA-vectors.
	
	In contrast to previous work on the Assignment Problem, we 
	look at the problem through the eyes of voting theorists. We come 
	up with several axiomatic and quantitative properties that are 
	desirable to fulfill by a mechanism if we assume that the Assignment 
	Problem is applied in the context of an election. 

	\section{Voting Rules}
	As we aim at selecting an individually-excellent line-up, a straightforward
	approach is to maximize the social welfare, which is determined by the 
	scores of
	the assigned candidate-position pairs. However, it is not always clear 
	which type of social welfare may be of interest. For example, the overall
	performance of a line-up may depend on the performance of the worst
	candidate. This may apply to team sports.
	Sometimes, however, the performance of a line-up is proportional to the sum
	of the scores and it does not hurt if some positions are not filled by a 
	qualified candidate. OWA-operators provide a convenient way to express both 
	these goals, 
	as well as a continuum of middle-ground approaches 
	\cite{DBLP:journals/tsmc/Yager88}.
	\smallskip

	\noindent\textbf{OWA-rules $\pmb{\mathcal{F}^\Lambda}$}. For a tuple 
	$a=(a_1,\dots,a_k)$ and $i\in [k]$, let $a[i]$ be the 
	$i$-th
	largest entry of $a$. We call $\Lambda:=(\lambda_1,\dots,\lambda_k)$ an 
	ordered weighed-average vector (OWA-vector) and define the ordered weighted 
	average of $a$ under $\Lambda$ as 
	$\Lambda(a_1,\dots,a_k):=\sum_{i\in [k]}\lambda_i \cdot 
	a[i]$~\cite{DBLP:journals/tsmc/Yager88}. 
	For a line-up $\pi$, we define $\Lambda(\pi)$ as the ordered weighted 
	average of the score vector of~$\pi$ 
	under $\Lambda$. That is, 
	$\Lambda(\pi):=\Lambda\big(\score_{p_1}(\pi),\dots, 
	\score_{p_q}(\pi)\big)$.
	The score of a line-up $\pi$ assigned by an OWA-rule $f^\Lambda$ is 
	$\Lambda(\pi)$.
	Rule $f^\Lambda$ chooses (possibly tied) line-ups with the highest
	score.

	Among this class of rules, we focus on the following four natural 
	ones, quite well studied in other contexts, such as finding a collective
	set of items 
	\cite{DBLP:conf/atal/AzizGGMMW15,conference-BFKKN20,DBLP:conf/aldt/ElkindI15,DBLP:journals/ai/SkowronFL16}:
	\begin{itemize}[itemsep=0pt,topsep=0pt]
		\item \emph{Utilitarian rule $f^{\ut}$}: $\Lambda^{\ut}:=(1,\dots,1)$. 
		This corresponds to computing a maximum weight matching 
		in the assignment graph. It is computable 
		in 
		$\mathcal{O}(m^3)$
		time for~$m$~candidates \cite{DBLP:books/daglib/p/Kuhn10}.
		\item \emph{Egalitarian rule $f^{\eg}$}: 
		$\Lambda^{\eg}:=(0,\dots,0,1)$. This corresponds to solving the Linear 
		Bottleneck
		Assignment Problem,
		which can be done in $\mathcal{O}(m^2)$ time for
		$m$~candidates~\cite{DBLP:journals/ior/Garfinkel71}.
		\item \emph{Harmonic rule $f^{\har}$}:
		$\Lambda^{\har}:=(1,\frac{1}{2},\frac{1}{3},\dots, \frac{1}{q})$. The
		computational complexity of finding a winning line-up under this rule 
		is open. In our experiments, we compute it using Integer 
		Linear Programming (ILP).
		\item \emph{Inverse harmonic rule $f^{\Ihar}$}: 
		$\Lambda^{\Ihar}:=(\frac{1}{q},\frac{1}{q-1},\dots,\frac{1}{2},1)$. 
		While computing the winning line-up for an arbitrary non-decreasing 
		OWA-vector is NP-hard~\cite{DBLP:journals/algorithmica/LescaMP19}, the 
		computational complexity of computing this specific rule is open. We 
		again use an ILP to compute a winning line-up.
	\end{itemize}
    
    \smallskip

	\noindent\textbf{Sequential-rules $\pmb{\mathcal{F}^{\seq}}$}. 
	OWA-rules require 
	involved algorithms and cannot be applied by hand 
	easily. In practice, humans
	tend to solve a line-up election in a simpler way, for instance, by 
	determining the election winners one by one. This procedure results in
	a class of quite intuitive sequential voting rules. A sequential rule is
	defined by some function $g$ that, given a line-up election
	$E=(\mathbf{S},C,P)$ and a set of already assigned candidates 
	$C_{\text{as}}$ and 
	positions $P_{\text{as}}$, returns the next position to be 
	filled. This 
	position is then filled by the remaining candidate $C\setminus 
	C_{\text{as}}$ with the 
	highest score on this position.  Sequentializing 
    decisions which partly depend on each other has also 
    proven to be useful in other voting-related problems, such as voting in 
    combinatorial domains 
    \cite{DBLP:journals/mss/LangX09}, or in the House Allocation problem in 
    form of the well-known mechanism of serial 
    dictatorship.
    We focus on the following three linear-time computable sequential rules.

	\smallskip \noindent\textit{Fixed-order rule} $f^{\seq}_{\fix}$. Here, the
	positions are filled in a fixed order (for simplicity, we assume the order 
	in which the positions appear in the election). The
	fixed-order sequential rule is probably the simplest way to generalize
	single-winner elections to the line-up setting and it enables us to make 
	the 
	decisions separately. Moreover, it is not necessary to 
	evaluate all candidates
	for all positions, which is especially beneficial if evaluating the
	qualification of a candidate on a position comes at some (computational) 
	cost.

	\smallskip \noindent\textit{Max-first rule} $f^{\seq}_{\ma}$. In the 
	max-first rule, at each step the position with the highest still 
	available score
	is filled. That is, 
	$g(\mathbf{S},C,P,C_{\text{as}},P_{\text{as}}):=\argmax_{p\in P \setminus 
	P_{\text{as}}} \max_{c\in C \setminus C_{\text{as}}} \score_p(c)$. 
	This is equivalent to adding at each step the remaining
	candidate-position pair with the highest score to
	the line-up. Max-first is intuitively appealing because a candidate who
	is outstanding at a position is likely to be assigned to it. Notably, this
	rule is an approximation of the utilitarian rule, as it corresponds to 
	solving
	the Maximum Weight Matching problem in the assignment graph by greedily 
	selecting the remaining edge with the
	highest weight.
	For every possible 
	tie-breaking, this approach is guaranteed to yield a
	$\frac{1}{2}$-approximation of the optimal solution in polynomial time~\cite{DBLP:journals/networks/Avis83}.

	\smallskip \noindent\textit{Min-first rule} $f^{\seq}_{\mi}$. In the 
	min-first rule, the position with the lowest score of the most-suitable 
	remaining candidate is filled next. That is, 
	$g(\mathbf{S},C,P,C_{\text{as}},P_{\text{as}}):=\argmin_{p\in P \setminus 
		P_{\text{as}}} \max_{c\in C \setminus C_{\text{as}}} \score_p(c).$
	The reasoning behind this is that the deciders focus first on filling 
	critical 
	positions where all candidates perform poorly.
	\begin{table}[t]
		\centering
		\setlength{\tabcolsep}{2pt}
		\renewcommand{\arraystretch}{1.2}
		\begin{tabular}{c|c|c|c|c|c|c|c|c|c} 
			&non & Pareto & reasonable & score & 
			\arxiv{position &}
			mono- 
			& 
			line-up \\ 
			&wasteful & optimal & satisfaction & consistent &\arxiv{ 
			consistent  
			& }
			tonicity & 
			enlargement \\
			\hline
			$f^{\ut}$  & S & S & $-$ & S& \arxiv{W &} S & 
			S \\
			$f^{\har}$  & S & S & $-$ & $-$ & \arxiv{$-$ &} $-$ &
			$-$ \\
			$f^{\Ihar}$  & S & S & $-$ & $-$ & \arxiv{$-$ &} $-$ 
			& $-$ \\
			$f^{\eg}$  & W\textsuperscript{\small $\dagger$} & 
			W\textsuperscript{\small $\dagger$} & $-$ & $-$ & 
			\arxiv{W &}
			$-$\textsuperscript{\small $\dagger$} & 
			W\textsuperscript{\small $\dagger$} \\
			$f^{\seq}_{\fix}$  & S & W & $-$ & W & \arxiv{W &} S &
			S \\
			$f^{\seq}_{\ma}$  & S & W & S & $-$ & \arxiv{W &} S &
			S \\
			$f^{\seq}_{\mi}$  & S & $-$ & $-$ & $-$ & \arxiv{W &} $-$ &
			$-$ \\
		\end{tabular}
		\caption{Overview of axiomatic properties of all studied rules. For 
		each of 
			the axioms, we indicate whether this axiom is strongly 
			satisfied (S), only weakly satisfied (W), or not satisfied at all 
			($-$). For the egalitarian rule, all entries marked with a 
			$\dagger$ can 
			be improved if a variant of this rule that selects the 
			egalitarian outcome with the highest summed score is 
			used.			
		}\label{ta:sum}
	\end{table}

	\section{Axiomatic Analysis of Voting Rules} \label{se::Axiom}
	
	We propose several axioms and properties that serve as a starting point to
	characterize and compare the introduced voting rules.
	We checked all introduced voting rules against all the axioms 
	and	collected the results in \autoref{ta:sum} (see 
	\autoref{se:proofs} for 
	all 
	proofs).\footnote{To convey an intuitive understanding of the 
	axioms, we 
	present in \autoref{se:intuition} for 
	each of them an undesirable example that might occur if a voting rule that 
	violates this axiom is used.} We will 
introduce two efficiency and one fairness axiom for line-ups. 
These definitions extend to voting rules as follows. A
voting rule $f$ \emph{strongly (weakly) satisfies} a given axiom if for 
each 
line-up 
election $E$, $f(E)$ contains only (some) line-ups satisfying the axiom.

	\smallskip \noindent \textbf{Efficiency axioms.} 	
	As our goal is to
	select individually-excellent outcomes, we aim at selecting line-ups
	in which the score of each position is as high as possible. Independent of
	conflicts between positions, there exist certain outcomes which are 
	suboptimal. For example, it is undesirable if there exists an unassigned 
	candidate
	that is more suitable for some position than the currently assigned 
	one.
	\begin{axiom}
		\textbf{Non-wastefulness}: In a line-up election 
		$(\mathbf{S},C,P)$,
		a line-up~$\pi$ is non-wasteful if there is no unassigned candidate 
		$c\notin \pi$ and a position 
		$p\in P$
		such that $\score_p(c)>\score_p(\pi)$. 
	\end{axiom}
	This axiom implies that in the special case of a single
	position, the candidate with the highest score for the position
	wins the election.	
	In the context of non-wastefulness, we only examine whether a
	line-up can be improved by assigning unassigned candidates. However, it is
	also possible to consider arbitrary rearrangements. This results in the 
	notion
	of score Pareto optimality.
	\begin{axiom}
		\textbf{Score Pareto optimality}: In a line-up election
		$(\mathbf{S},C,P)$, a line-up $\pi$ is score Pareto dominated
		if there exists a line-up $\pi'$ such that for all $p\in P$,
		$\score_p(\pi')\geq \score_p(\pi)$ and there exists a position $p\in P$ 
		with
		$\score_p(\pi')> \score_p(\pi)$. A line-up is score Pareto
		optimal if it is not score Pareto dominated.
	\end{axiom}
	While all OWA-rules with an OWA-vector containing no zeros are clearly 
	strongly non-wasteful and strongly score Pareto optimal, 
	the 
	egalitarian rule satisfies 
	both axioms only weakly, as this rule selects a line-up purely based on its 
	minimum score. All sequential rules naturally satisfy strong
	non-wastefulness. Yet, by breaking ties in a suboptimal way, all 
	of them may output line-ups that are not score Pareto optimal. 
	While for the fixed-order rule and the max-first rule at least one winning 
	outcome is always score Pareto optimal, slightly counterintuitively, there 
	exist 
	instances where the 
	min-first rule does not output any score Pareto optimal line-ups.
	 
	\smallskip
	\noindent\textbf{Fairness axioms.}
	Another criterion to judge the quality of a voting rule is to 
	assess 
	whether positions and candidates are treated in a fair way.
	The underlying assumption for fairness in this context is that every 
	position should have the best possible candidate assigned. 
	Similarly, from
	candidates' perspective, one could argue that each candidate 
	deserves to be assigned to the position for which the candidate is most
	suitable. In the following, 
	we call a candidate or position for which fairness is violated 
	dissatisfied. Unfortunately, line-ups in which all positions and candidates 
	are satisfied simultaneously may not
	exist. 
	That is why we consider a restricted fairness property, where we call a 
	candidate-position 
	pair $(c,p)\in \pi$ reasonably dissatisfied in $\pi$ if
	candidate $c$ has a higher score for $p$ than $\pi_p$ and $c$ is either 
	unassigned or $c$'s score 
	for $p$ is higher than for $\pi(c)$.\footnote{For a more extensive 
		discussion of reasonable 
		satisfaction, see \autoref{se:rdiss}.}
	
	\begin{axiom}
		\textbf{Reasonable satisfaction}: A line-up $\pi$ is 
		reasonably satisfying if there are no two positions $p$ and $p'$ such 
		that 
		$$\score_p(\pi_{p'})> 
		\score_p(\pi_{p}) \text{ and } 
		\score_{p}(\pi_{p'})> \score_{p'}(\pi_{p'})$$ and there is no 
		candidate 
		$c\notin \pi$ and position~$p$ such that 
		$\score_p(c)>\score_p(\pi_p).$
	\end{axiom}

	\noindent  
	It is straightforward to prove that all winning line-ups
	under the 
	max-first rule are reasonably satisfying; hence, a reasonably
	satisfying outcome always exists. Note that it is also possible to motivate 
	reasonable satisfaction as a 
	notion of stability if we assume that candidates and positions are allowed 
	to leave their currently assigned partner to pair up with a new candidate 
	or position. Thus, reasonable satisfaction resembles the notion of 
	stability in the context of the Stable Marriage problem 
	\cite{DBLP:journals/tamm/GaleS13}. 
	
	However, fulfilling reasonable satisfaction may come at the cost of 
	selecting a 
	line-up that 
	is suboptimal for various notions of social welfare. For some $\epsilon>0$, 
	consider an election with $P=\{p_1,p_2\}$, $C=\{a,b\}$, 
	$\score_{p_1}(a)=2$, $\score_{p_2}(a)=\score_{p_1}(b)=2-\epsilon$, and 
	$\score_{p_2}(b)=0$.
	We see that the outcome $(b,a)$ of utilitarian welfare $4-2\epsilon$ 
	maximizes utilitarian social welfare, 
	while the outcome $(a,b)$ of utilitarian welfare $2$ is 
	the only outcome 
	satisfying reasonable satisfaction. Therefore, it is interesting to measure 
	the 
	price of 
	reasonable satisfaction in terms of utilitarian welfare. Analogous to the 
	price 
	of stability \cite{DBLP:journals/siamcomp/AnshelevichDKTWR08}, we define 
	this as the maximum utilitarian social 
	welfare 
	achievable by a reasonably satisfying outcome, divided by the
	maximum 
	achievable utilitarian social welfare. The example from above already 
	implies that this price is upper bounded by $\frac{1}{2}$. In fact, this 
	bound is 
	tight, as the max-first rule that only outputs reasonably satisfying 
	line-ups is a $\frac{1}{2}$-approximation of 
	the utilitarian outcome.  
	
	\smallskip
	\noindent \textbf{Voting axioms.}
	We now formulate several axioms that are either closely related to axioms 
	from 
	single-winner~\cite{DBLP:reference/choice/Zwicker16} or multi-winner 
	voting~\cite{DBLP:journals/scw/ElkindFSS17}. We define all axioms
	using two parts, 
	(a) and (b). On a high level, condition~(a) imposes that certain 
	line-ups should be winning after modifying a line-up election, while 
	condition~(b) demands 
	that no other line-ups become (additional) winners after the
	modifications. If a voting rule only fulfills condition~(a), then we say 
	that it \emph{weakly satisfies} the corresponding axiom. If
	it fulfills both conditions, then we say that it \emph{strongly satisfies} 
	the
	corresponding axiom.
	
	In single-winner and multi-winner voting, the consistency 
	axiom  requires that if a voting rule selects the same outcome in two 
	elections (over the same candidate set), then this 
	outcome is also winning in the combined election \cite{YOUNG197443}. We 
	consider a 
	variant of this axiom, adapted to our setting.
	\begin{axiom}
		\textbf{Score consistency}: For two line-up elections 
		$(\mathbf{S},C,P)$ and 
		$(\mathbf{S}',C,P)$ with
		$f(\mathbf{S},C,P)\cap 
		f(\mathbf{S}',C,P)\neq \emptyset$ it holds 
		that: 
		a) $f(\mathbf{S},C,P)\cap f(\mathbf{S}',C,P)\subseteq 
		f(\mathbf{S}+\mathbf{S}',C,P)$ and 
		b) $f(\mathbf{S},C,P)\cap f(\mathbf{S}',C,P)\supseteq 
		f(\mathbf{S}+\mathbf{S}',C,P).$
	\end{axiom}
	The utilitarian rule is the only 
	OWA-rule that satisfies weak 
	(and even strong) score consistency. The 
	fixed-order rule is 
	the only other considered rule that satisfies weak score consistency. For 
	all other rules, it is possible to construct simple 
	two-candidates two-positions line-up elections where this axiom is 
	violated.  

	We now consider a different way of combining two 
	line-up elections which is special to our setting. In line-up elections, it 
	is possible to combine the 
	set of considered positions: Imagine that we split a line-up election  into 
	two sub-elections with the same candidate set in both elections and 
	partly overlapping sets of 
	positions. Assume further that the considered voting rule outputs a winning 
	outcome of the fist sub-election and a winning outcome of the second 
	sub-election such that all shared positions are filled by the same 
	candidates 
	and all non-shared positions by different candidates that do not appear in 
	the other line-up at all. Now, considering the full election on all 
	positions, the union of 
	these outcomes should then be a winning outcome. Formally, we say that 
	two line-ups,~$\pi'$ defined on $P'$ and 
	$\pi''$ defined on $P''$, are 
	\textit{overlapping-disjoint} if they assign the same candidates on all
	positions from $P'\cap P''$ and a disjoint set of candidates to the other 
	positions. For two
	overlapping-disjoint line-ups $\pi'$ and~$\pi''$, we write $\pi'\cup \pi''$ 
	to
	denote the line-up~$\pi^*$ combining the two, that is, for all $p\in 
	P'\colon
	\pi^*_p=\pi_p$ and for all~$p\in P''\colon \pi^*_p=\pi'_p$. 
	\begin{axiom}
		\textbf{Position consistency}: Let
		$(\mathbf{S},C,P)$ be a line-up election and $P',P''\subseteq P$ two 
		subsets of positions with $P'\cup P''=P$.  
		If there exist two overlapping-disjoint outcomes~$\pi'\in 
		f(\mathbf{S},C,P')$ and $\pi''\in f(\mathbf{S},C,P'')$, then it holds 
		that 
		(a) $\pi'\cup \pi''\in f(\mathbf{S},C,P)$ and 
		(b) for all $\pi^*\in f(\mathbf{S},C,P)$,
		there exist 
		overlapping-disjoint $\pi'\in 
		f(\mathbf{S},C,P')$ and $\pi''\in f(\mathbf{S},C,P'')$ such that 
		$\pi^*=\pi'\cup \pi''$. 
	\end{axiom} 
	It is easily possible to construct small instances for all considered 
	voting 
	rules where strong position consistency is violated. However, all rules 
	except the two 
	harmonic rules satisfy the axiom in the weak sense. The harmonic rules fail 
	the axiom because it is possible to exploit the fact that by extending the 
	election, the coefficient by 
	which a score in the line-up is multiplied changes compared to the two 
	sub-elections.
	
	Besides focusing on consistency related considerations, it is also 
	important to examine how the winning line-ups change if 
	the election 
	itself is  modified. We start by considering a variant of monotonicity 
	\cite{DBLP:journals/dam/Fishburn82}.
	
	\begin{axiom}
		\textbf{Monotonicity}: Let $(\mathbf{S},C,P)$ be a line-up 
		election with a
		winning line-up $\pi$. Let $(\mathbf{S}',C,P)$ be the line-up 
		election obtained from $(\mathbf{S},C,P)$ 
		by increasing $\score_{p}(\pi_p)$ for some $p$.
		Then, it 
		holds that (a) $\pi$ is still a winning line-up, that is, $\pi \in 
		f(\mathbf{S}',C,P)$, 
		and (b) no new winning line-ups are created, that is, for all $\pi'\in 
		f(\mathbf{S}',C,P)$ 
		it holds that~$\pi'\in f(\mathbf{S},C,P)$.
	\end{axiom}
	
	While the utilitarian, fixed-order, and max-first rule all 
	satisfy 
	strong monotonicity, all other rules fail even 
	weak monotonicity. As these negative results are intuitively surprising, 
	we present their proof here. 
	
	\begin{proposition} \label{th::monot}
		The egalitarian rule $f^{\eg}$, the harmonic rule $f^{\har}$, and the 
		min-first rule~$f^{\seq}_{\min}$ all violate weak 
		monotonicity.
	\end{proposition}
	\begin{proof}
		We present counterexamples for all three rules: 
		\begin{center}
			$E_1:$ \begin{tabular}{c|c|c} 
				
				& $p_1$ & $p_2$ \\
				\hline
				$a$ & $0$ & $3$  \\
				$b$ & $3$ & $0$ \\
				$c$ & $4$ & $0$  \\
			\end{tabular}
			\qquad
			$E_2:$
			\begin{tabular}{c|c|c|c} 
				
				& $p_1$ & $p_2$ & $p_3$\\
				\hline
				$a$ & $4$ & $1$ & $0$ \\
				$b$ & $4.75$ & $3$ & $0$ \\
				$c$ & $0$ & $0$ & $2$ \\
			\end{tabular}
			\qquad 
			$E_3:$ 
			\begin{tabular}{c|c|c} 
				& $p_1$ & $p_2$ \\
				\hline
				$a$ & 2 & 1 \\
				$b$ & 0 & 0 
			\end{tabular}
		\end{center}
		\noindent In $E_1$, $(b,a)$ and $(c,a)$ are winning under $f^{\eg}$. 
		However, after 
		increasing $\score_{p_2}(a)$ to~$4$,~$(c,a)$ has an egalitarian score 
		of $4$ and thereby becomes the unique winning line-up. 
		We now turn to election $E_2$ and voting rule $f^{\har}$. It is clear
		that 
		$c$ will be assigned to $p_3$ in every outcome. In fact, $(a,b,c)$ is 
		the 
		winning line-up in $E_2$, as~$\Lambda^{\har}(a,b,c)=1\cdot 
		4+\frac{1}{2}\cdot 
		3+\mathbf{\frac{1}{3}\cdot 2}>1\cdot 4.75+\frac{1}{3}\cdot 
		1+\mathbf{\frac{1}{2}\cdot 2}
		=\Lambda^{\har}(b,a,c)$. After increasing $\score_{p_3}(c)$ to $3$, 
		$(b,a,c)$ 
		becomes the unique winning line-up in $E_2$, 
		as~$\Lambda^{\har}(a,b,c)=1\cdot 4+\frac{1}{2}\cdot 
		3+\mathbf{\frac{1}{3}\cdot 3}<1\cdot 4.75+\frac{1}{3}\cdot 
		1+\mathbf{\frac{1}{2}\cdot 3}
		=\Lambda^{\har}(b,a,c)$. By the modification, the score of $(b,a,c)$ 
		increases more than the score of $(a,b,c)$, because in $(b,a,c)$ 
		$\score_{p_3}(c)$ is multiplied by a larger coefficient.
		Lastly, 
		considering $f^{\seq}_{\min}$, $(b,a)$ is the unique winning outcome in 
		$E_3$. 
		After increasing $\score_{p_2}(a)$ to 3, the ordering in which the 
		positions get assigned changes, and thereby, $(a,b)$ becomes the unique 
		winning 
		line-up.
	\end{proof}

	In the context of multi-winner voting, an additional 
	monotonicity axiom is sometimes considered: Committee enlargement 
	monotonicity deals with the
	behavior of the set of winning outcomes if the size of the committee is
	increased~\cite{DBLP:journals/scw/BarberaC08,DBLP:journals/scw/ElkindFSS17}.
	 We 
	generalize this axiom to 
	our
	setting in a straightforward way. 
	\begin{axiom}
		\textbf{Line-up enlargement monotonicity}: Let 
		$(\mathbf{S},C,P)$ be a line-up election 
		and let $(\mathbf{S}',C,P')$ with $P'=P\cup \{p^*\}$ be an election 
		where  
		position $p^*$ and the scores of 
		candidates for this position have been added. It holds that (a) for 
		all  
		$\pi\in f(\mathbf{S},C,P)$ there exists some $\pi'\in 
		f(\mathbf{S}',C,P')$ such that
		$\pi_P\subset\pi'_{P'}$ and (b) for all $\pi'\in f(\mathbf{S}',C,P')$ 
		there exists 
		some $\pi\in f(\mathbf{S},C,P)$ such that $\pi_P\subset \pi'_{P'}$.  
	\end{axiom} 
	Note that line-up enlargement monotonicity does not require that the 
	selected 
	candidates are assigned to the same position in the two outcomes $\pi$ and 
	$\pi'$.
	Despite the fact that this axiom seems to be very natural, neither of the 
	two
	harmonic rules satisfy it at all. Intuitively, the reason for this is that 
	by 
	introducing 
	a new position, the coefficients in the OWA-vector ``shift''. 
	Moreover, surprisingly, the min-first rule also violates the weak version 
	of this axiom. The proof for this consists of a rather involved 
	counterexample (see~\autoref{th:LEMM} in the appendix) exploiting 
	the fact 
	that by 
	introducing a new position, the 
	order in 
	which the positions are filled may change. All other rules, apart from the 
	egalitarian one, satisfy strong 
	line-up enlargement monotonicity; the egalitarian rule 
	satisfies line-up enlargement monotonicity only in the weak sense. We 
	conclude with presenting the proof 
	that 
	the utilitarian rule, 
	$f^{\ut}$, satisfies weak line-up enlargement monotonicity, as the proof 
	nicely 
	illustrates how it is possible to reason about this axiom.

	\begin{proposition} \label{th:utilLEM}
		The utilitarian rule $f^{\ut}$ satisfies weak line-up enlargement 
		monotonicity.
	\end{proposition}
	\begin{proof}
		Let $\pi$ be a winning line-up 
		of the 
		initial election $E=(\mathbf{S},C,P)$ and let $\pi'$ be a winning 
		line-up of the extended 
		election $E'=(\mathbf{S}',C,P\cup \{p^*\})$ such that there exists a 
		candidate $c\in C$ with
		$c\in \pi$ and $c\notin 
		\pi'$. We claim that it is always possible to construct 
		from 
		$\pi'$ a winning line-up $\pi^*$ of the extended election such that 
		all candidates from $\pi$ appear in $\pi^*$: Initially, we set 
		$\pi^*:=\pi'$. As long 
		as there 
		exists a candidate $c\in C$ with~$c\in \pi$ and $c\notin \pi^*$, we set 
		$\pi^*_{\pi(c)}:=c$. Let $\widetilde{P}$ be the 
		set of all positions  
		where $\pi'$ and~$\pi^*$ differ. Note that none 
		of the replacements can change the candidate assigned to~$p^*$. Thus, 
		it holds 
		that $\widetilde{P}\subseteq P$. 
		
		Obviously, all candidates 
		from $\pi$ appear in $\pi^*$. For the 
		sake 
		of contradiction, let us assume that $\pi^*$ is not a winning line-up 
		of the extended election. 
		Consequently, the summed score of $\pi^*$ has decreased by the 
		sequence of 
		replacements described above, which implies that the summed scores of 
		candidates on positions from $\widetilde{P}$ is higher in $\pi'$ than 
		in~$\pi$:~$\Lambda^{\ut}(\pi|_{\widetilde{P}})<\Lambda^{\ut}(\pi'|_{\widetilde{P}})$.
		We claim that using this assumption, it is possible to modify~$\pi$ 
		such 
		that its utilitarian score increases, which leads to a contradiction, 
		as we have assumed that~$\pi$ is a winning line-up. Let
		$\pi_{\textrm{alt}}$ be a line-up resulting from copying~$\pi$ and then
		replacing all candidates assigned to positions in $\widetilde{P}$ by the
		candidates assigned to these positions in $\pi'$. This is possible as 
		$p^*
		\notin \widetilde{P}$. By our assumption, $\pi_{\textrm{alt}}$ has a 
		higher
		utilitarian score than $\pi$. It remains to argue that 
		$\pi_{\textrm{alt}}$ is
		still a valid outcome, that is, every candidate is only assigned to at 
		most
		one position. This directly follows from the observation that 
		if~$p\in\widetilde{P}$ with $\pi'_p\in \pi$, then at some point during 
		the
		construction of~$\pi^*$,~$\pi'_p$ is kicked out of the line-up and is 
		assigned
		to position~$\pi(\pi'_p)$ at some later point, which implies 
		that~$\pi(\pi'_p)\in \widetilde{P}$.
	\end{proof}
	
	\noindent\textbf{Summary.}  From an axiomatic perspective, the utilitarian 
	rule is probably the
	most appealing one. Indeed, it satisfies all axioms except weak reasonable 
	satisfaction, which imposes quite rigorous restrictions on every rule 
	fulfilling it. For the egalitarian rule, although this rule is pretty 
	simple, 
	both 
	efficiency axioms are only weakly satisfied and, slightly surprisingly, 
	score 
	consistency and monotonicity are not satisfied at all. As 
	in~\autoref{th::monot}, most of the 
	counterexamples for the egalitarian rule utilize that the OWA-vector of 
	this rule contains some zeros. If one adapts the egalitarian rule such 
	that the egalitarian outcome with the highest summed score is chosen, 
	strong non-wastefulness, 
	strong score
	Pareto optimality, strong monotonicity, and strong line-up enlargement
	monotonicity are additionally satisfied. Note that this variant of the 
	egalitarian rule is still 
	computable in polynomial time.

	Both harmonic rules are less 
	appealing from an axiomatic perspective because they do not satisfy any of 
	our 
	considered voting axioms. On a 
	high level, as illustrated in 
	\autoref{th::monot}, this is because the corresponding OWA-vectors consist 
	of multiple 
	different entries. Thereby, some modifications change the coefficients by 
	which the scores are multiplied in some undesirable way. Overall, the 
	contrast 
	between the utilitarian rule and the harmonic 
	rule is quite remarkable because they come both from the same class and 
	work pretty similarly.
	
	Turning to sequential rules, apart from reasonable dissatisfaction,  
	the fixed-order rule outperforms the other two. However, a clear 
	disadvantage of the fixed-order rule is that a returned winning line-up
	might not be score Pareto optimal. 
	As the max-first rule fulfills all axioms---except score consistency---at 
	least 
	weakly, and it is the only voting rule that is reasonably satisfying,
	this rule is also
	appealing if satisfaction of the candidates or positions 
	is an important criterion. The min-first rule, on the other hand, does not 
	weakly satisfy any axiom except 
	non-wastefulness.
	The reason for this is that
	in some elections modifying the election changes the order in 
	which the positions are filled.
	The considerable differences between the max-first and min-first rule are 
	quite surprising, as they first appear to be symmetric.

	\section{Experiments}
	In this section, we analyze the proposed rules experimentally.
	We first describe how we generated 
	our data, and then present and analyze the results.

	\smallskip
	\noindent \textbf{FIFA data.} In the popular video game FIFA 19 
	\cite{fifa}, 18.207 
	soccer players have their own avatar. To mimic the quality of a player, 
	experts 
	assessed 
	them on 29 attributes, 
	such as sprint speed, shot power, agility, and heading \cite{dataset}. From 
	these, the game 
	computes the quality of a player on each possible position, such as left 
	striker or right wing-back etc., in a soccer 
	formation, using a weighted sum of the attribute 
	scores with coefficients depending on the position in 
	question \cite{fifaPos}. We used this data to 
	model 
	a coach of a national team~$\mathcal{X}$ 
	that wants to find an ``optimal'' assignment of players with
	nationality~$\mathcal{X}$ to positions in a formation he came up with. This 
	can be modeled as a line-up election.
	
	In soccer, there exist several possible formations consisting of different 
	positions a team can play in. We 
	fixed one formation, that is,
	a set of ten different positions (see~\autoref{fi:SocLine} in the appendix 
	for a visualization).

	In the generated elections,
	the candidates are some selected number
	of players of a given nationality with the highest summed score,
	the positions are the field positions in a selected soccer formation,
	and the scores are those assigned by FIFA 19 for the player
	playing on a particular position. 
	We considered 
	84 national teams (those with over ten~field players in FIFA 19).
	
	\smallskip
	\noindent\textbf{Synthetic data.} 
	We also generated a synthetic dataset (M2) consisting of $1.000$~line-up 
	elections. Here, every candidate $c$ has a 
	ground qualification $\mu_c \in [0.4,0.7]$ and every position $p$ has a 
	difficulty $\alpha_p\in[1,2]$ both drawn uniformly at random. For each 
	candidate $c$ and for each position $p$, we sample a basic 
	score~$\beta_{c,p}$ 
	from a Gaussian distribution with mean $\mu_c$ and standard deviation 
	$0.05$. 
	The score of a 
	candidate-position pair is then calculated as: 
	$\score_p(c)=\beta_{c,p}^{\alpha_p}$. The intuition 
	behind this is that very talented candidates are presumably not strongly 
	affected by the difficulty of a position, whereas weaker candidates may 
	feel completely overburdened by a difficult position. In addition, 
	we generated another synthetic dataset (M1) that resembles the FIFA data. 
	The description of this model and the results for it are deferred to 
	\autoref{se:M1} and \autoref{fi:missres}.
	
	For both models, we normalized each line-up 
	election by dividing all scores by the maximum score of a 
	candidate-position pair.
	
	\smallskip
	\noindent \textbf{Analysis of experimental results.}
	We focus on the case of ten candidates on ten 
	positions, as this is the most relevant scenario in the 
	FIFA setting. However, we also conducted experiments for twenty 
	candidates on ten 
	positions and twenty candidates on twenty positions, where we observed the 
	same trends as in 
	the case presented here (see \autoref{se:moreResults} for 
	diagrams). For 
	settings with more candidates than positions, 
	however, 
	the 
	differences between the rules are less visible.	In the following, we refer 
	to the (possibly 
	invalid) 
	outcome 
	where every position gets 
	assigned its best candidate as the \emph{utopic}
	outcome.
	To visualize our 
	results, we use violin 
	plots. In a violin plot, the white dot represents the median, the thick bar 
	represents
	the interquartile range, and the thin line represents the range of the data 
	without 
	outliers. Additionally, a distribution interpolating the 
	data is plotted on both sides of the center.	
	
	\smallskip
	\noindent\textit{Comparison of data models.} To compare the datasets, 
	we calculated different metrics designed to measure the amount of 
	``competition'' 
	in instances. For example, we
	calculated the difference between the summed score of the utopic 
	outcome and
	the summed 
	score of a 
	utilitarian outcome (see
	\autoref{se:CompModels} for details). Generally speaking, 
	the M2 model produces instances with more ``competition'' than the FIFA 
	data which
	helps us to make the 
	differences between the
	rules more pronounced.

	\begin{figure}[t] 
		{\centering 
		\includegraphics[width=0.95\textwidth]{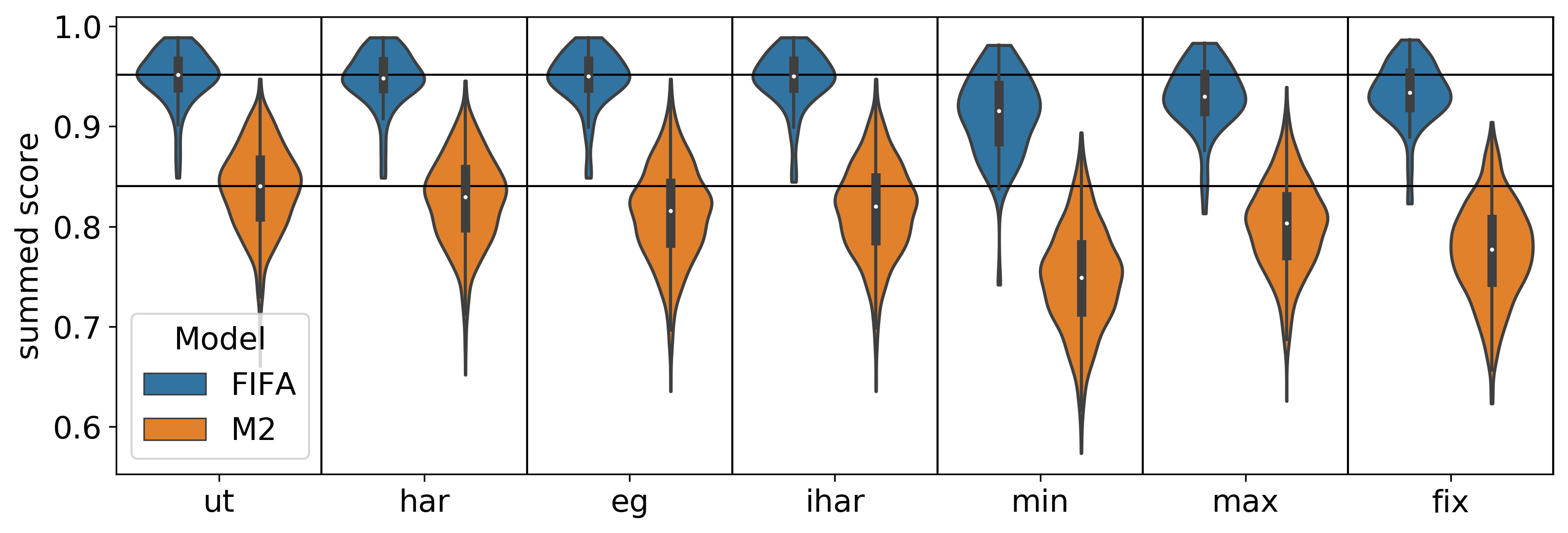}
		
		}
		\caption{Summed score of winning line-ups for different voting rules. 
		The black horizontal lines indicate the 
		median for the utilitarian rule on the two datasets.} 
		\label{fi:SS}
	\end{figure}
	\begin{figure}[t]
			\includegraphics[width=0.95\textwidth]{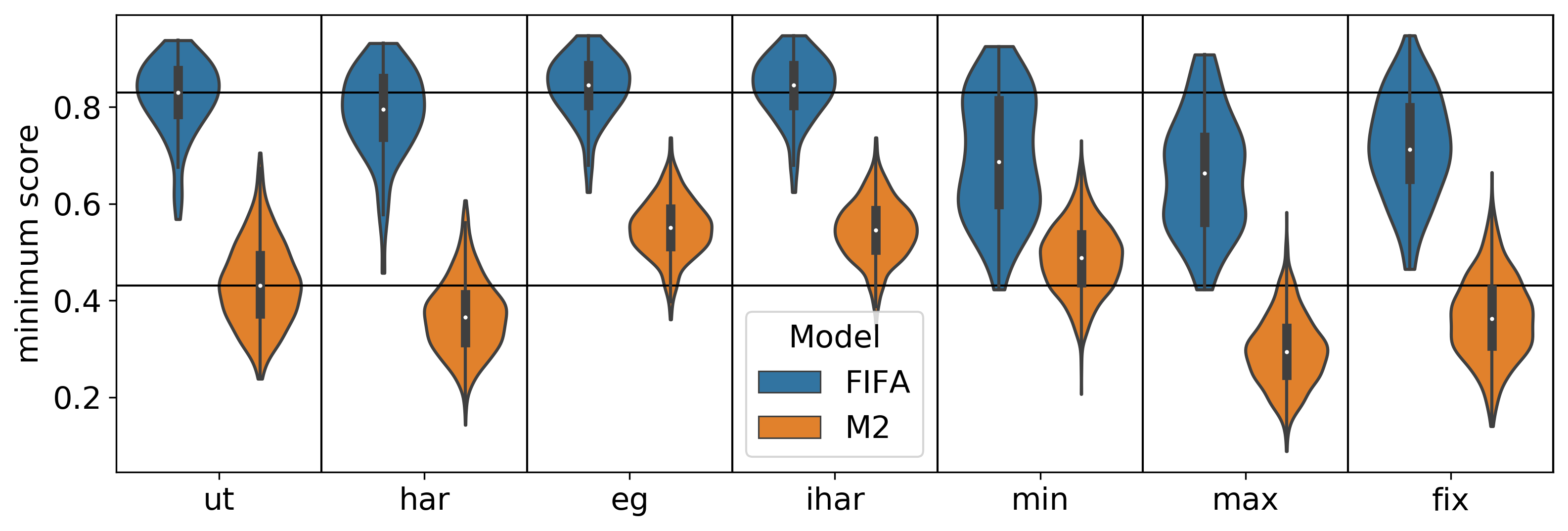}
			\caption{Minimum score in winning line-ups for different voting 
			rules.} 
			\label{fi:min} 
	\end{figure}
	
	\smallskip
	\noindent\textit{Comparison of voting rules.} We compare the different 
	voting 
	rules by examining the following four metrics:
	\begin{enumerate*}[label=(\roman*)]
	 \item the summed score of the computed winning line-up $\pi$ 
	 normalized by the summed score of the utopic
	  outcome,
	 \item the minimum score of a position in
	  the winning line-up,
	 \item the Gini coefficient\footnote{The Gini coefficient is a metric to 
	 measure the dispersion of a probability distribution; it is zero for 
	 uniform distributions and one for distributions with a unit step 
	 cumulative distribution function (see \autoref{se:Gini} for a 
	 formal 
	 definition)}.  
	 of the score vector, and
	 \item the amount of reasonable 
	  dissatisfaction measured as the sum of all reasonable dissatisfactions, 
	  that
	  is, the difference between the score of a position $p$ in $\pi$ and the 
	  score
	  of a better candidate~$c$ on $p$ if the candidate-position pair $(p,c)$ 
	  is 
	  reasonably dissatisfied.
	\end{enumerate*} 
	For the egalitarian rule, if multiple line-ups 
	are winning, then we always select the line-up with the highest summed 
	score. 
	
	Concerning the summed score (see \autoref{fi:SS}), as expected, all 
	four OWA-rules clearly outperform the three sequential rules. The 
	OWA-rules 
	all behave remarkably similar, especially on the FIFA data. The 
	utilitarian rule
	produces by definition line-ups with the highest possible summed score, 
	closely
	followed by the harmonic rule. Turning to
	the sequential rules, the min-first rule produces the worst
	results, while the max-first rule produces the best results.

	\begin{figure}[t]
		{\centering 
		\includegraphics[width=0.95\textwidth]{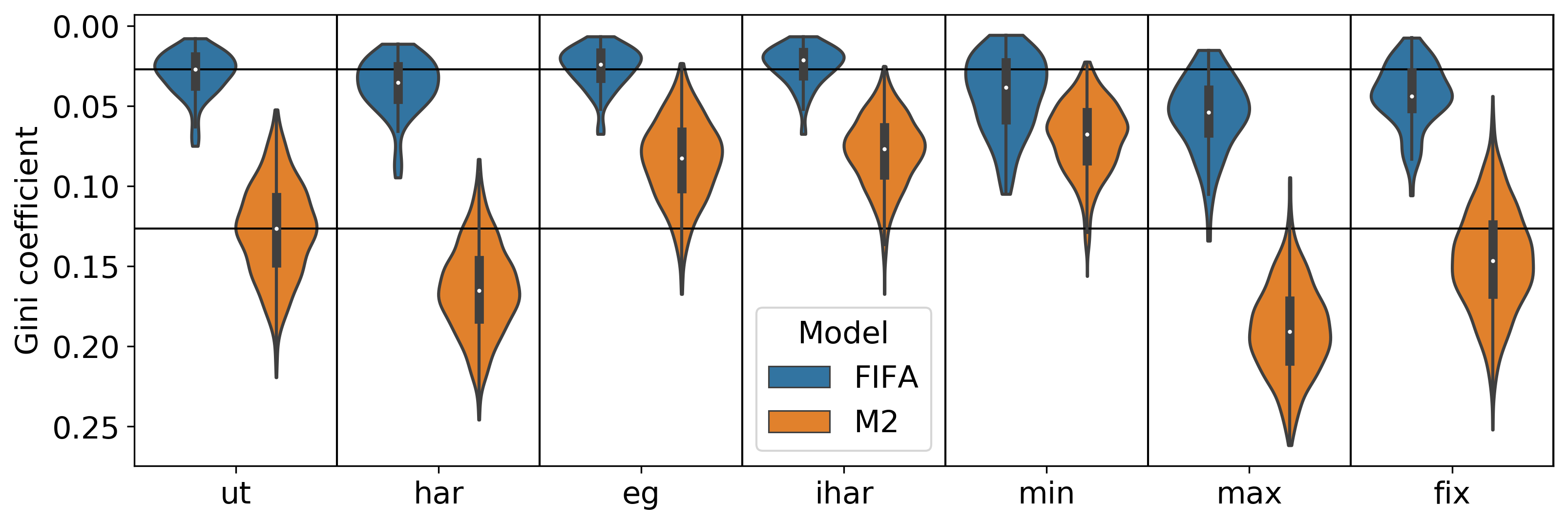}
		
		}
		\caption{Gini coefficient of score vector of line-ups 
		for different voting rules.} \label{fi:gini}
	\end{figure}
	
	Turning to the minimum score (see \autoref{fi:min}),
	the OWA-rules mostly outperform the 
	sequential rules. The two rules with the highest minimum score 
	are 
	the egalitarian rule and the inverse harmonic rule. The utilitarian and 
	harmonic 
	rule produce slightly worse results on the easier FIFA data and 
	significantly worse results on the more demanding M2 data. Among 
	the 
	sequential rules, the min-first rule performs best, 
	sometimes even outperforming the utilitarian rule, while the max-first rule 
	produces the worst results. For the Gini coefficient (see 
	\autoref{fi:gini}), the overall picture is 
	quite similar, i.e., rules that produce line-ups with a higher minimum 
	score 
	also produce line-ups that are more balanced in general.

	\begin{figure}[t] 
		{\centering 
		\includegraphics[width=0.95\textwidth,height=4.3cm]{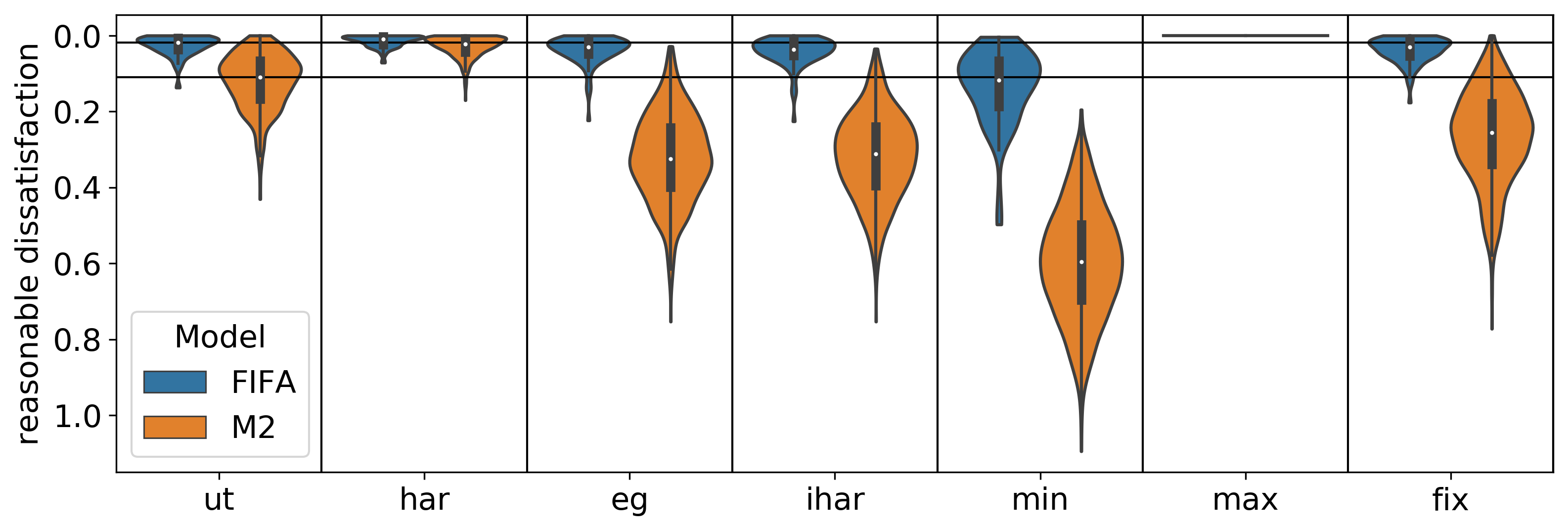}
		
		}
		\caption{
		Reasonable dissatisfaction in winning line-ups for 
		different voting 
		rules.} 
		\label{fi:RDiss}
	\end{figure} 
	Lastly, considering the amount of reasonable dissatisfaction (see 
	\autoref{fi:RDiss}), by 
	definition, the max-first rule does not produce any reasonable 
	dissatisfaction. Among the other rules, the harmonic 
	rule produces the best results followed by the utilitarian rule. 
	The egalitarian, inverse harmonic, and fixed-order
	rule all produce around 
	double the amount 
	of reasonable dissatisfaction, while the min-first rule produces 
	significantly worse results by an additional factor of two.

	\noindent\textbf{Summary.} Somewhat surprisingly, all 
	OWA-rules outperform all 
	sequential rules for all quantities,
	with only two exceptions: The min-first rule produces pretty balanced 
	outcomes and the max-first rule produces no reasonable dissatisfaction.
	However, even if one aims at optimizing mainly one of these two 
	quantities, it is usually recommendable to use an OWA-rule. Selecting 
	the inverse harmonic rule instead of the min-first rule
	results in outcomes which are comparably balanced, have significantly higher
	summed and minimum scores, and have way less reasonable dissatisfaction. 
	Using the harmonic rule instead of the max-first rule will
	introduce only little reasonable dissatisfaction in exchange for more 
	balanced
	line-ups with higher summed and minimum scores. Comparing the different 
	OWA-rules to each other, it is possible to differentiate the utilitarian and
	harmonic rule on the one side, from the egalitarian and inverse harmonic
	rule (which behave particularly similarly) on the other side: Rules from 
	the 
	former class tend to favor more 
	imbalanced line-ups with lower minimum but higher summed score
	and less reasonable dissatisfaction.
	
	
	\section{Discussion}
	Overall, the considered OWA-rules 
	produce better outcomes than the sequential rules. 
	Nevertheless,  
	sequential rules might sometimes be at an advantage, since sequential rules 
	are, generally speaking, more transparent, 
	more intuitive, and easier to explain.
	If one requires a sequential rule, 
	either the 
	fixed-order rule or the
	max-first rule should be chosen, as the min-first rule 
	violates nearly all studied axioms and produces undesirable outcomes. 
	Focusing on OWA-rules, the 
	harmonic and inverse harmonic rules are rather to be avoided, as they fail 
	to fulfill all considered voting and fairness axioms. Comparing the 
	utilitarian and egalitarian rule, from an axiomatic perspective, the 
	utilitarian rule is at an advantage, because it satisfies the most axioms 
	among 
	all considered 
	rules. However, choosing between these two rules in practice should depend 
	on the application, as the line-ups produced by these rules maximize 
	different metrics. A rule that could somehow incorporate
	egalitarian \emph{and}  
	utilitarian considerations is the product rule, which selects 
	outcomes with the highest product
	of scores.

	For future work, it would be interesting to look at line-up elections 
	that take as input the preferences of voters instead of aggregated scores. 
	Analogously to multi-winner voting, new rules for this setting could, for 
	example, focus on selecting proportional and diverse, instead of 
	individually-excellent, line-ups. Such a path would also 
	require developing appropriate
	axioms. Another possible line of future research could be to run 
	experiments using preference data to examine the influence of the selected 
	single-winner voting rule to aggregate the preferences into 
	scores on the selected line-up. There are also several algorithmic problems 
	that arise from our work. For instance, the computational complexity of 
	computing a winning outcome under the (inverse) harmonic rule and even more 
	generally of computing winning outcomes for arbitrary non-increasing 
	OWA-rules is open.
	
	
	\bibliographystyle{splncs04.bst}
	\bibliography{bib}
	\newpage
	
	\appendix 
	\section{Axiomatic Analysis}
	\subsection{Reasonable Dissatisfaction}\label{se:rdiss}
	Recall that it is possible to model the following situation as a line-up 
	election: A company wants to 
	fill different 
	positions in different teams of the company and publishes an open call for 
	applications to which several candidates respond. The task is then to 
	assign the candidates to the positions in the company. In such a setting, 
	each group 
	offering a job wants to get the best 
	candidate for the job and is unsatisfied if this is not 
	the 
	case. Similarly, each candidate wants to be assigned to the position she is 
	most suitable for. The reasoning behind this is that presumably every 
	candidate prefers to do tasks for which she is qualified and wants to
	make most out of herself.  Unfortunately, line-ups which 
	assign all 
	candidates 
	to their best position and each position its best candidate may not
	exist. 
	That is why it is necessary to differentiate between different types of
	dissatisfaction. For example, the dissatisfaction of a candidate who is way
	less suitable for all positions than all other candidates is hard to
	circumvent. In contrast, imagine there exists a
	candidate $c$ who is unsatisfied, as she is either not assigned or assigned 
	to
	a position for which she is less qualified than for another position $p$ for
	which, even more, she is more suitable than the current candidate 
	filling~$p$.
	Then, the dissatisfaction of~$c$ and the dissatisfaction of $p$ are quite
	reasonable and voters might agree that assigning~$c$ to her current position
	treated her and position $p$ unfairly. This setup shall motivate and 
	illustrate our definition of
	reasonable dissatisfaction presented in the main body of the paper. 
	Interestingly, every winning line-up under the
	max-first rule is reasonable satisfying, which proves the following 
	proposition.
	 
	\begin{proposition} \label{th::rSat}
		In every line-up election, an outcome without reasonable 
		dissatisfaction is guaranteed to exist.
	\end{proposition}
	\begin{proof} 
		Let $E$ be a line-up election. We claim that all winning line-ups  
		$\pi$ of $E$ under the max-first rule are 
		reasonably 
		satisfying, which proves the proposition. Assume that there exist two 
		positions $p$ and $p'$ 
		fulfilling the first 
		condition 
		for reasonable dissatisfaction. Then, if 
		$p$ has been filled by the rule before $p'$, then it needs to 
		hold that $\score_p(\pi_p)\geq \score_p(\pi_{p'})$, thereby 
		contradicting 
		the 
		condition. If $p'$ were filled by the algorithm before~$p$, then it 
		would 
		need to
		hold that $\score_{p'}(\pi_{p'})\geq \score_p(\pi_{p'})$, thereby 
		contradicting
		the condition. No candidate $c\notin \pi$ and position~$p$ can fulfill 
		the
		second condition of reasonable satisfaction, as otherwise $c$ would 
		have been
		assigned to $p$. 
	\end{proof}
	
	\subsection{Intuitive Explanations of Axioms} 
	\label{se:intuition}
	In the following, we try to give a more intuitive understanding 
	of the considered axioms by interpreting them in the language of our 
	introductory soccer example. Here, the candidates 
	are 
	the players of the team and positions are positions in a soccer formation. 
	Scores of candidate-positions pairs are derived from approvals of the 
	coaching staff. 
	For all 
	axioms, we present a plausible drawback that might occur if this axiom is 
	violated.
	\begin{description}
		\item[Non-wastefulness] After agreeing on a 
		line-up, the coaches realize that there exists a position and an 
		unassigned
		player such that the coaches agree that he is more suitable to fill this
		position than the currently assigned player.
		\item[Score Pareto optimality] After agreeing on a line-up, 
		the coaches realize that there exists a different line-up where the
		coaches agree that, taking all positions into account, this 
		line-up is better than the one which they decided for.
		\item[Score consistency] The team of 
		defensive coaches 
		meets and agrees on a line-up. At the same time, the team of 
		offensive 
		coaches meets and agrees independently on the 
		same 
		line-up. Afterwards, all coaches meet together and agree on a 
		line-up that is different from the one that each set of coaches 
		came up with independently. 
		\item[Position consistency] The defensive 
		coaches agree on a line-up of defenders and midfielders. The 
		offensive 
		coaches agree on a line-up of strikers and midfielders. Both 
		line-ups 
		coincide on the midfielders and no player who is placed as a 
		defender 
		in the first line-up is placed as a striker in the second line-up. 
		However, in a joint meeting, the coaches decide on a full line-up 
		that is different from the line-up of the defensive coaches on the 
		defenders or midfielders 
		(despite the fact that it would have been possible to
		combine the two initial line-ups).
		\item[Monotonicity] The coaches agree on a 
		line-up for a game. In this 
		game, the center-forward plays very well, so afterwards more 
		coaches believe that he is suitable to fill this position. All 
		other opinions remain unchanged. In the next game, some other line-up 
		is chosen.
		\item[Line-up enlargement monotonicity] In a training session on a 
		smaller field, the coaches select their best line-up consisting out 
		of 
		five players. The day later, for the next normal game with eleven 
		players, one of these 
		players is benched.
		\item[Reasonable satisfaction] The team has a star player 
		where 
		everyone agrees that he is a perfect center-forward and better than 
		everyone else on this position. In the 
		current line-up, he is placed in the midfield and very unsatisfied 
		and demotivated by this.
	\end{description}
	
	\subsection{Missing Proofs} \label{se:proofs}
	In the following, we provide for all considered rules and axioms proofs 
	whether the rule satisfies the axiom strongly, weakly, or not at all. 
	We split these proofs into two parts: We start by examining the OWA-rules, 
	before we turn to sequential rules. Moreover, within 
	each section, because often similar ideas are required, we group the 
	results by the axioms. For each axiom, the proofs for the 
	different rules from the considered class can be found next to each other. 
	For an overview of the results, we refer to \autoref{ta:sum} from the 
	main body of the paper. We present the different axioms in the same order 
	as they appear in~\autoref{ta:sum}.
	\subsubsection{OWA-rules}
	We only 
	consider non-negative normalized OWA-vectors $\Lambda$, 
	that is,
	vectors in which all entries are non-negative and the largest entry is one. 
	It is easily possible to normalize arbitrary OWA-vectors by dividing all 
	entries by the maximum entry in the vector. 
	To explain how we will reason about OWA-rules in the following, let us look 
	at the following line-up election:
	\begin{center}
		\begin{tabular}{c|c|c|c} 
			
			\diagbox{$C$}{$P$} & $p_1$ & $p_2$ & $p_3$\\
			\hline
			$a$ & $1$ & $0$ & $8$ \\
			$b$ & $2$ & $4$ & $5$ \\
			$c$ & $0$ & $6$ & $7$ \\
		\end{tabular}
	\end{center}
	Consider the line-up $(a,b,c)$. It consists of one 
	candidate-position pair with score~$1$, namely~$a$ who is assigned to 
	$p_1$, one candidate-position pair with score $4$, namely $b$ who is 
	assigned to $p_2$, and one candidate-position pair of score $7$ namely, $c$ 
	who is assigned to~$p_3$. Now, imagine that we want to compute the score 
	assigned to this outcome by the harmonic rule, that is, 
	$f^{\Lambda^{\har}}$ with $\Lambda^{\har}=(1,\frac{1}{2},\frac{1}{3})$. The 
	score is the sum of the highest score of a candidate-position pair plus 
	$\frac{1}{2}$ times the second highest score of a candidate-position pair 
	plus 
	$\frac{1}{3}$ times the lowest score of a candidate-position pair. So for 
	the line-up~$(a,b,c)$ it is $7+\frac{1}{2}*4+\frac{1}{3}*1=11 \frac{1}{3}$. 
	We denote the score of $(a,b,c)$ under the harmonic rule 
	as~$\Lambda^{\har}(a,b,c)=\frac{1}{3}*1+\frac{1}{2}*4+7$. Note that we sort 
	the terms of the sum by the ordering of the candidates in the line-up. For 
	instance, the score of the line-up $(b,a,c)$ under 
	the harmonic rule is 
	$\Lambda^{\har}(b,a,c)=\frac{1}{2}*2+\frac{1}{3}*0+1*7.$ We start by 
	examining non-wastefulness.

	\begin{theorem}
		All OWA-rules $f^\Lambda$ are weakly non-wasteful and weakly score 
		Pareto 
		optimal. For an OWA vector $\Lambda$ with only strictly positive 
		entries,  
		$f^\Lambda$~is strongly non-wasteful and score Pareto optimal.
	\end{theorem}
	\begin{proof}
		First of all, note that weak score Pareto optimality implies weak 
		non-wastefulness. To prove that all OWA-rules $f^\Lambda$ are weakly 
		score Pareto optimal, for the sake of contradiction,  let us assume 
		that $f^\Lambda$ has returned only score 
		Pareto dominated line-ups. Let $\pi$ be one of these line-ups,
		which is score Pareto dominated by some line-up $\pi'$. 
		However, this implies that $\Lambda(\pi')\geq \Lambda(\pi)$,
		meaning that $\pi'$ is also selected as a winning line-up.

		Note that for all $\Lambda$ with strictly positive entries, it even 
		holds that $\Lambda(\pi')> 
		\Lambda(\pi)$ 
		and thereby that only score Pareto optimal line-ups are selected.
		However, for OWA-vectors~$\Lambda$ with a zero entry at position $i$, 
		$f^\Lambda$ is not strongly score Pareto optimal. To see this, consider 
		an election where all scores are different natural numbers. Let $\pi$ 
		be a line-up in which some position $p\in P$ has the $i$-th highest 
		score. However, then the outcome $\pi'$ in which 
		each position except $p$ gets assigned the same candidate and $p$ gets 
		assigned 
		a candidate $c$ with $\score_{p}(c)=\score_{p}(\pi)-1$ is still a 
		winning line-up but obviously score Pareto dominated by $\pi'$.
		
		The same arguments can also be used to prove the statement for 
		non-wastefulness.
	\end{proof}
	
	\begin{theorem}
		For all OWA-vectors $\Lambda\neq (1,0,\dots,0)$, $f^{\Lambda}$ is not 
		weakly 
		reasonable satisfying. For $\Lambda=(1,0,\dots,0)$, 
		$f^{\Lambda}$ is weakly but not strongly reasonably satisfying.
	\end{theorem}
	\begin{proof}
		Here and in some of the following proofs, we prove the two 
		negative parts only for OWA-vectors of size two. However, it is easily 
		possible to generalize the provided counterexamples to arbitrary OWAs 
		of size $k>2$ by inserting $k-2$
		positions and, for each position, a designated candidate. For each 
		inserted position, all candidates except the designated candidate have 
		a high negative score, while the designated candidate has a 
		score which is either above every score used in the example or slightly 
		below 
		every score used in the example. By selecting how many of the $k-2$ 
		designated candidates have a score above or below the scores from the 
		example, it is possible to 
		select the window in the OWA-vector the example accounts for. 
		
		Let us consider the following two-candidates two-positions line-up 
		election.
		Let $\Lambda$ be an arbitrary length-two OWA-vector without zero 
		entries and
		let $0<x<1$ be the smaller of the two entries of $\Lambda$ (i.e.,\,we 
		have
		either~$\Lambda = (1,x)$ or~$\Lambda = (x,1)$; we will treat the two 
		cases
		$(1,0)$ and $(0,1)$ separately). Let us consider the following line-up
		election:
		\begin{center}
			\begin{tabular}{c|c|c|} 
				
				\diagbox{$C$}{$P$}  & $p_1$ & $p_2$ \\
				\hline
				$a$ & $1$ & $1+\frac{y}{2} $ \\
				$b$ & $0$ & $1$
			\end{tabular}
		\end{center}
		with~$y = x$. In this election, outcome $(a,b)$ has score $1+x$ under 
		OWA-vector 
		$\Lambda$, while outcome~$(b,a)$ depending 
		on whether~$\Lambda$ is decreasing or increasing has score either 
		~$1+\frac{x}{2}$ or~$x+\frac{x\cdot x}{2}$. Consequently,~$(a,b)$ is 
		selected 
		as the unique winner by all $f^\Lambda$ with $\Lambda\neq (1,0)$ and 
		$\Lambda\neq (0,1)$. 
		However, in the outcome $(a,b)$, the candidate-position pair $(a,p_2)$ 
		is reasonably dissatisfied, as $a$'s score for position~$p_2$ is higher 
		than $b$'s score for position~$p_2$ and higher than $a$'s score for 
		position~$p_1$. 
		
		For $\Lambda=(0,1)$, we can use the election from above with arbitrary 
		strictly positive~$y$. In any case, $f^{\Lambda}$ will select $(a,b)$ 
		as the unique winning line-up. In this line up, the candidate-position 
		pair $(a,p_2)$ is reasonably dissatisfied.
		
		To prove that $f^{(1,0,\dots)}$ is weakly reasonable satisfying, 
		let $\pi$ be a winning line-up of some line-up election under 
		$f^{(1,0,\dots)}$. Let $(c,p)$ be the candidate-position 
		pair 
		with the highest score in $\pi$. It follows that this is the highest 
		score $c$ can achieve and the highest score $p$ can achieve, as 
		otherwise $\pi$ would not be a winning line-up. Moreover, all 
		line-ups~$\pi'$ with $\pi_p=\pi'_p$ are also winning line-ups. One of 
		them is
		guaranteed to be reasonable satisfying, as all line-ups provided by the
		max-first rule are included in them. Clearly, due to the zero
		entries, one can easily construct examples where some of these winning 
		line-ups
		cause reasonable dissatisfaction.
	\end{proof}
	
	\begin{theorem}
		For all OWA-vectors $\Lambda\neq (1,\dots 1)$, $f^\Lambda$ is not 
		weakly score consistent. The utilitarian rule, $f^{\ut}$, is strongly 
		score 
		consistent.
	\end{theorem}
	\begin{proof}
		To prove the first part, let us consider the following two-candidates 
		two-positions line-up
		election that
		can be extended to an arbitrary number of candidates and positions: 
		\begin{center}
			$E$: \begin{tabular}{c|c|c} 
				 \diagbox{$C$}{$P$} & $p_1$ & $p_2$ \\
				\hline
				$a$ & 4 & 3 \\
				$b$ & 2 & 1
			\end{tabular}
			\qquad
			$E'$:
			\begin{tabular}{c|c|c} 
				\diagbox{$C$}{$P$} & $p_1$ & $p_2$ \\
				\hline
				$a$ & 1 & 3 \\
				$b$ & 2 & 4
			\end{tabular}
			\qquad
			$E^*$: 
			\begin{tabular}{c|c|c} 
				\diagbox{$C$}{$P$} & $p_1$ & $p_2$ \\
				\hline
				$a$ & 5 & 6 \\
				$b$ & 4 & 5
			\end{tabular}
		\end{center}
		For $f^{(1,x)}$ with $x\in [0,1)$, in the first election $E$, it holds 
		that 
		$\Lambda(a,b)=4+x>2x+3=\Lambda(b,a)$. 	
		In the second election, it holds that 
		$\Lambda(a,b)=x+4>2x+3=\Lambda(b,a)$. 
		However, in the combined election $E^*$, where we summed up the 
		score 
		matrices of the first two elections, it holds that 
		$\Lambda(a,b)=5+5x<4x+6=\Lambda(b,a)$. 
		Hence, weak score consistency is 
		violated. 
		
		For $f^{(x,1)}$ with $x\in [0,1)$, in the first election, it holds that 
		$\Lambda(a,b)=4x+1<2+3x=\Lambda(b,a)$  
		in the first election. In the second election, it holds that 
		$\Lambda(a,b)=1+4x<2+3x=\Lambda(b,a)$. 
		However, in the combined election $E^*$, where we summed up the 
		score 
		matrices of the first two elections, it holds that
		$\Lambda(a,b)=5x+5>4+6x=\Lambda(b,a)$, which violates score consistency.
		
		It remains to prove that $f^{\ut}$ is strongly score 
		consistent. 
		Let $\pi$ be some winning line-up in the two original elections 
		$E=(\mathbf{S},C,P)$ and $E'=(\mathbf{S}',C,P)$. We show
		that~$\pi$ is a winner in~$E^*$. For the 
		sake of 
		contradiction, 
		let us assume that there exists a line-up~$\pi'$ with a higher 
		utilitarian score than $\pi$ in the combined 
		election $E^*=(\mathbf{S}+\mathbf{S}',C,P)$. 
		However, as it holds that the score of a line-up in~$E^*$
		is equal to the sum of the scores of this line-up in~$E$ and~$E'$, this 
		implies that for at least one of the two 
		initial elections 
		it also needs to 
		hold that the score 
		of~$\pi'$ is higher than the score of~$\pi$. However, this contradicts 
		the assumption 
		that $\pi$ is a winning line-up for both $E$ and $E'$. 
		
		To prove that the utilitarian rule also satisfies condition (b) of 
		score consistency, let~$\pi$ be some winning line-up $\pi$ 
		of the
		combined election $E^*$ that is also a winning line-up in both
		$E$ and $E'$ (such~$\pi$ must exist due to the above argument for the
		first direction).
		Let $\Lambda^{\ut}_{E}(y)$ denote the utilitarian score
		of a line-up $y$ in
		an election $E$. Let~$\pi' \neq \pi$ be another line-up that is
		also winning in the combined election (if such a line-up exists). We 
		show that~$\pi'$ is
		a winner in~$E$ and~$E^*$ as well. Since both~$\pi'$ and~$\pi$ win 
		in~$E^*$,
		it holds that $\Lambda^{\ut}_{E^*}(\pi)=\Lambda^{\ut}_{E^*}(\pi')$. 
		Moreover,
		as $\pi$ wins in $E$ and $E'$, $\Lambda^{\ut}_{E}(\pi)\geq
		\Lambda^{\ut}_{E}(\pi')$ and~$\Lambda^{\ut}_{E'}(\pi)\geq 
		\Lambda^{\ut}_{E'}(\pi')$. 
		For every line-up~$x$, we have that 
		$\Lambda^{\ut}_{E^*}(x)=\Lambda^{\ut}_{E}(x)+\Lambda^{\ut}_{E'}(x)$,
		which implies that 
		$\Lambda^{\ut}_{E}(\pi)= \Lambda^{\ut}_{E}(\pi')$ and 
		$\Lambda^{\ut}_{E'}(\pi)= 
		\Lambda^{\ut}_{E'}(\pi')$. Hence, $\pi'$ is a winning line-up in $E$ 
		and $E'$.
	\end{proof}
	
	\begin{theorem} 
		The harmonic rule, $f^{\har}$, and the inverse harmonic rule, 
		$f^{\Ihar}$, violate weak position consistency. The egalitarian rule, 
		$f^{\eg}$, 
		and the utilitarian rule, $f^{\ut}$, are weak but not strong position 
		consistent.
	\end{theorem}
	\begin{proof}
               \underline{Harmonic rule $f^{\har}$}:
	    Let $P'=\{p_1,p_2\}$ and $P''=\{p_2,p_3\}$ and 
		let 
		$E^*=(\mathbf{S},C,P'\cup P'')$ be the following line-up election: 
		\begin{center}
			\begin{tabular}{c|c|c|c} 
				
				 \diagbox{$C$}{$P$} & $p_1$ & $p_2$ & $p_3$\\
				\hline
				$a$ & $1$ & $3$ & $0$ \\
				$b$ & $3$ & $4.1$ & $0$ \\
				$c$ & $0$ & $0$ & $5$ \\
			\end{tabular}
		\end{center}
		In the first sub-election, $E'=(\mathbf{S},C,P')$, it holds that 
		$\Lambda^{\har}(a,b)=\frac{1}{2}\cdot 1+1\cdot 4.1>1\cdot 
		3+\frac{1}{2}\cdot 3=\Lambda^{\har}(b,a)$. Thus, $(a,b)$ is the unique 
		winning 
		line-up.
		In the second sub-election, $E''=(\mathbf{S},C,P'')$, $(b,c)$ is the 
		unique 
		winning 
		line-up. 
		However, in the full election~$E^*$, $\Lambda^{\har}(a,b,c)= 
		\frac{1}{3}\cdot 1+\frac{1}{2}\cdot 4.1+1\cdot 5<\frac{1}{2}\cdot 
		3+\frac{1}{3}\cdot 3+1\cdot 
		5=\Lambda^{\har}(b,a,c)$.
		Thereby,~$(a,b,c)$ is not a winning line-up, which violates weak 
		position 
		consistency.
		
		\medskip
		\noindent \underline{Inverse harmonic rule $f^{\Ihar}$}:
		Let 
		$P'=\{p_1,p_2\}$ and $P''=\{p_2,p_3\}$ and let 
		$E^* = (\mathbf{S},C,P'\cup P'')$ be the following line-up election: 
		\begin{center}
			\begin{tabular}{c|c|c|c} 
				
				\diagbox{$C$}{$P$} & $p_1$ & $p_2$ & $p_3$\\
				\hline
				$a$ & $2$ & $1$ & $0$ \\
				$b$ & $3 \frac{3}{4}$ & $2$ & $0$ \\
				$c$ & $0$ & $0$ & $\frac{1}{2}$ \\
			\end{tabular}
		\end{center}
		In the first sub-election, $E'=(\mathbf{S},C,P')$, $(a,b)$ is the 
		unique 
		winning line-up, 
		as~$\Lambda^{\Ihar}(a,b)=\frac{1}{2}\cdot 2+1\cdot 2>\frac{1}{2}\cdot 3 
		\frac{3}{4}+1\cdot 1=\Lambda^{\Ihar}(b,a)$. 
		In the second sub-election,~$E''=(\mathbf{S},C,P'')$, $(b,c)$ is the 
		unique 
		winning line-up. 
		However, in the full 
		election~$E^*$,~$\Lambda^{\Ihar}(a,b,c)=\frac{1}{3}\cdot 
		2+\frac{1}{2}\cdot 2+1\cdot 
		\frac{1}{2}< 
		\frac{1}{3}\cdot 3 
		\frac{3}{4}+\frac{1}{2}\cdot 1+1\cdot 
		\frac{1}{2}=\Lambda^{\Ihar}(b,a,c)$.
		Thereby,~$(a,b,c)$ is not a winning line-up, which violates weak 
		position 
		consistency.
		
		\medskip
		\noindent \underline{Egalitarian rule $f^{\eg}$}:  We start by proving 
		that this rule 
		satisfies weak position consistency. 
		Assume that $E^*=(\mathbf{S},C,P)$ is a line-up election and 
		$P',P''\subseteq P$ with $P'\cup P'' = P$. Let~$\pi'$ be a winning 
		line-up of the election $E'=(\mathbf{S},C,P')$ and $\pi''$ be a winning 
		line-up of the 
		election $E''=(\mathbf{S},C,P'')$, where $\pi'$ and $\pi''$ are 
		overlapping-disjoint. We claim that $\pi'\cup \pi''$ is then a winning 
		line-up of the full election $E^*$. First of all, note that for all 
		outcomes~$\pi^*$ of the full election $E^*$ it holds that 
		$\min(\Lambda^{\eg}(\pi^*|_{P}),\Lambda^{\eg}(\pi^*|_{P'}))=\Lambda^{\eg}(\pi^*)$.
		Thereby,
		the 
		existence 
		of an outcome of the combined election $\pi^*$ with higher score than 
		$\pi'\cup 
		\pi''$ implies that either 
		$\Lambda^{\eg}(\pi^*|_{P'})>\Lambda^{\eg}(\pi')$, which 
		contradicts the 
		assumption that $\pi'$ is a winning line-up of $E'$, or 
		$\Lambda^{\eg}(\pi^*|_{P''})>\Lambda^{\eg}(\pi'')$, which 
		contradicts the assumption that
		$\pi''$ is a winning line-up of $E''$. 
		
		To prove that $f^{\eg}$ does not satisfy strong position consistency, 
		we exploit the property that not all winning line-ups are score Pareto 
		optimal. Consider the following line-up election 
		$E^*=(\mathbf{S},C,P)$: 
		\begin{center}
			\begin{tabular}{c|c|c|c} 
				
				\diagbox{$C$}{$P$} & $p_1$ & $p_2$ & $p_3$\\
				\hline
				$a$ & $1$ & $0$ & $0$ \\
				$b$ & $0$ & $3$ & $2$ \\
				$c$ & $0$ & $2$ & $3$ \\
			\end{tabular}
		\end{center}
		In the sub-election on $P'=\{p_1,p_2\}$, $(a,b)$ is the unique winning 
		line-up. On 
		$P''=\{p_2,p_3\}$, $(b,c)$ is the unique winning line-up. However, in 
		the full 
		election, $(a,c,b)$ is also a winning line-up, which violates strong 
		position consistency. 
		
		\medskip
		\noindent \underline{Utilitarian rule $f^{\ut}$}: We prove that 
		$f^{\ut}$ satisfies weak but not strong position 
		consistency.  Let again 
		$E^*=(\mathbf{S},C,P)$ 
		be a line-up elections and 
		$P',P''\subseteq P$ with $P'\cup P'' = P$. In the following, we write 
		$L$ to 
		denote $P'\setminus P''$, $M$ to denote $P'\cap P''$ and 
		$R$ to 
		denote $P''\setminus P'$. 
		
		Assume that $\pi'$ is a winning line-up of the first sub-election 
		$E'=(\mathbf{S},C,P')$ on $P'$ 
		and~$\pi''$ of the second sub-election $E''=(\mathbf{S},C,P'')$ on 
		$P''$,
		where $\pi'$ and $\pi''$ are 
		overlapping-disjoint. We claim that $\pi^*=\pi'\cup \pi''$ is then a 
		winning line-up 
		of the 
		full election $E^*$.
		
		Assume, for the sake of contradiction, that there exists a line-up 
		$\widetilde{\pi}^*$ of $E^*$ with a higher utilitarian score than 
		$\pi^*$. Let $\widetilde{\pi}':=\widetilde{\pi}^*|_{P'}$ and 
		$\widetilde{\pi}'':=\widetilde{\pi}^*|_{P''}$. 
		As 
		$\pi'$ and $\pi''$ are winning line-ups, it needs to hold that 
		$\Lambda^{\ut}(\pi')\geq 
		\Lambda^{\ut}(\widetilde{\pi}')$ and $\Lambda^{\ut}(\pi'')\geq 
		\Lambda^{\ut}(\widetilde{\pi}'')$. 
		From 
		this it 
		follows that 
		$$\Lambda^{\ut}(\widetilde{\pi}^*|_L)+ 
		\Lambda^{\ut}(\widetilde{\pi}^*|_M)\leq 
		\Lambda^{\ut}(\pi^*|_L)+ 
		\Lambda^{\ut}(\pi^*|_M) \text{ and }$$ 
		$$\Lambda^{\ut}(\widetilde{\pi}^*|_M)+ 
		\Lambda^{\ut}(\widetilde{\pi}^*|_R)\leq \Lambda^{\ut}(\pi^*|_M)+ 
		\Lambda^{\ut}(\pi^*|_R).$$ Hence,
		\begin{equation}\label{eq:in}
			\Lambda^{\ut}(\widetilde{\pi}^*|_L)+ 
			2\cdot \Lambda^{\ut}(\widetilde{\pi}^*|_M)+	
			\Lambda^{\ut}(\widetilde{\pi}^*|_R)\leq 
			\Lambda^{\ut}(\pi^*|_L)+2\cdot \Lambda^{\ut}(\pi^*|_M)+	
			\Lambda^{\ut}(\pi^*|_R)
		\end{equation}		
		
		\noindent Recall that, assuming that $\widetilde{\pi}^*$ has a higher 
		score than 
		$\widetilde{\pi}$, it needs to hold that:
		\begin{equation}\label{eq:in2}
		\begin{aligned}
			\Lambda^{\ut}(\widetilde{\pi}^*) 
			=&\Lambda^{\ut}(\widetilde{\pi}^*|_L) +
			\Lambda^{\ut}(\widetilde{\pi}^*|_M) +	
			\Lambda^{\ut}(\widetilde{\pi}^*|_R)> \\
			&\Lambda^{\ut}(\pi^*|_L)+\Lambda^{\ut}(\pi^*|_M)+	
			\Lambda^{\ut}(\pi^*|_R)=\Lambda^{\ut}(\pi^*).
		\end{aligned}
		\end{equation}
		The only possibility that both inequality \eqref{eq:in} and 
		\eqref{eq:in2} hold at the same time is 
		that~$\Lambda^{\ut}(\pi^*|_M)-\Lambda^{\ut}(\widetilde{\pi}^*|_M)=x$ 
		and 
		$\Lambda^{\ut}(\widetilde{\pi}^*|_L)-\Lambda^{\ut}(\pi^*|_L)=y$ and 
		$\Lambda^{\ut}(\widetilde{\pi}^*|_R)-\Lambda^{\ut}(\pi^*|_R)=y$
		for some $x,y,z>0$ with 
		$x<y+z$. 
		
		For the sake of contradiction, assume that an outcome 
		$\widetilde{\pi}^*$ that fulfills the thee conditions specified above 
		exists. We now compare $\pi^*$ and 
		$\widetilde{\pi}^*$ and argue why this is not possible. Let us first 
		look 
		at the 
		differences of $\pi^*$ and $\widetilde{\pi}^*$ on $L$: Because of the 
		optimality 
		of $\pi'$, it is not possible to increase the 
		score of $\pi^*$ on~$L$ by using candidates who are not already 
		assigned 
		in $\pi^*$ to positions from $M$. Thereby, in~$\widetilde{\pi}^*|_L$ 
		some candidates from~$\pi^*|_{L}$ must have been 
		replaced by 
		some candidates from $\pi^*|_{M}$. From the optimality of $\pi'$ 
		it follows that by modifying~$L$ as described above, an increase of
		$\Lambda^{\ut}(\pi^*|_{L})$ by 
		some~$x$ always results in a decrease of at least $x$ on 
		$\Lambda^{\ut}(\pi^*|_{M})$. Moreover, 
		it is 
		never 
		possible to increase the score on~$R$ simultaneously by such 
		rearrangements, as 
		all candidates that might get free are candidates from $R$ which 
		are---due to 
		the optimality of $\pi''$--- not usable to increase the score on $L$. 
		As 
		the 
		same argument also holds for $R$, it is never possible to increase the 
		scores 
		as required and thereby $\pi^*$ is always a winning line-up.  
		
		We consider the following line-up election $E^*=(\mathbf{S},C,P)$ to 
		prove that $f^{\ut}$ 
		violates strong position consistency:
		\begin{center}
			\begin{tabular}{c|c|c|c} 
				
				\diagbox{$C$}{$P$}& $p_1$ & $p_2$ & $p_3$\\
				\hline
				$a$ & $0$ & $0$ & $0$ \\
				$b$ & $1$ & $1$ & $0$ \\
				$c$ & $0$ & $0$ & $0$ \\
			\end{tabular}
		\end{center}
		In the sub-election on $P'=\{p_1,p_2\}$, $(a,b)$ is a winning outcome. 
		On~$P''=\{p_2,p_3\}$,~$(b,c)$ is winning outcome. However, in the full 
		election $E^*$,
		$(b,a,c)$ is a winning outcome. However $(a,c)$ is not a winning 
		outcome of the 
		second sub-election on $P''$.
	\end{proof}
	
	\begin{theorem}
		The utilitarian rule, $f^{\ut}$, satisfies strong monotonicity. The 
		egalitarian rule,~$f^{\eg}$, violates weak 
		monotonicity. For all OWA-vectors $\Lambda$ 
		with 
		at least 
		three succeeding non-zero entries of which the last two are unequal, 
		$f^\Lambda$ 
		does not 
		satisfy weak monotonicity. 
	\end{theorem}
	\begin{proof}
		\underline{Utilitarian rule $f^{\ut}$}: We show that $f^{\ut}$ 
		satisfies strong
		monotonicity.  Let $\pi$ be a winning line-up 
		of some 
		line-up election $(\mathbf{S},C,P)$. Then, by 
		increasing 
		$\score_p(c)$ 
		for some $(c,p)\in C\times P$ with $c=\pi_p$ by some arbitrary~$x$, the 
		score of 
		$\pi$ 
		is increased by~$x$. Moreover, 
		the 
		score of all other line-ups is also increased by at most $x$. Thereby, 
		$\pi$ 
		remains a winning line-up, and no new winning line-ups 
		can be created, as every line-up that is winning in the modified 
		elections needs to have the same utilitarian score as $\pi$ before the 
		modification. 
		
		\medskip
		\noindent \underline{Egalitarian rule $f^{\eg}$}: Consider the 
		following line-up 
		election: \begin{center}
			\begin{tabular}{c|c|c} 
				
				\diagbox{$C$}{$P$} & $p_1$ & $p_2$ \\
				\hline
				$a$ & $0$ & $3$  \\
				$b$ & $3$ & $0$ \\
				$c$ & $4$ & $0$  \\
			\end{tabular}
		\end{center}
		Under the egalitarian rule, $(b,a)$ is a winning line-up of this 
		election. However, after increasing the score of $a$ on $p_2$ to four, 
		$(b,a)$ is no longer a winning line-up, as now $(c,a)$ is the unique 
		winning line-up.
		
		\medskip
		\noindent \underline{(Inverse) harmonic rule $f^{\har}$/$f^{\Ihar}$}: 
		Finally, we prove that for all OWA-vectors~$\Lambda$ 
		with 
		at least 
		three succeeding non-zero entries of which the last two are not equal, 
		$f^\Lambda$ 
		does not 
		satisfy weak monotonicity. Note that this also includes the harmonic 
		rule and inverse harmonic rule. 
		Let $\Lambda=(x,y,z)$. We assume that $x\neq 0$ and $y\neq z$. Consider 
		the following line-up election:
		\begin{center}
			\begin{tabular}{c|c|c|c} 
				
				\diagbox{$C$}{$P$} & $p_1$ & $p_2$ & $p_3$\\
				\hline
				$a$ & $4$ & $1$ & $0$ \\
				$b$ & $4+\frac{0.5y+1.5z}{x}$ & $3$ & $0$ \\
				$c$ & $0$ & $0$ & $2$ \\
			\end{tabular}
		\end{center}
		Clearly, in all winning line-ups, $c$ is assigned to $p_3$. Thereby, 
		the two 
		line-ups which could be winning are $(a,b,c)$ and $(b,a,c)$ with scores
		$\Lambda(a,b,c)=4x+3y+2z$ and $\Lambda(b,a,c)=x\cdot 
		(4+\frac{0.5y+1.5z}{x})+1z+2y =4x+2.5y+2.5z$. We now 
		increase $\score_{p_3}(c)$ to $3$. The scores of $(a,b,c)$ and 
		$(b,a,c)$ 
		change 
		as follows: $\Lambda(a,b,c)=4x+3y+3z$ and $\Lambda(b,a,c)=4x+3.5y+2.5z$.
		Assuming now that $y>z$, $(a,b,c)$ is the unique winning line-up in 
		the 
		original election. However, in the altered election, $(b,a,c)$ is 
		the unique 
		winning line-up. This violates weak monotonicity. 
		Assuming $z>y$, $(b,a,c)$ is the unique winning line-up in the 
		original 
		election, while $(a,b,c)$ is the unique winning line-up in the altered 
		election. This violates weak monotonicity.
	\end{proof}
	
	\begin{theorem}
		The egalitarian rule, $f^{\eg}$, satisfies weak line-up enlargement 
		monotonicity but fails to 
		satisfy 
		strong line-up enlargement monotonicity. The utilitarian rule, 
		$f^{\ut}$, satisfies strong 
		line-up enlargement monotonicity. The harmonic rule, $f^{\har}$, and 
		inverse harmonic rule, $f^{\Ihar}$, both 
		fail weak 
		line-up enlargement monotonicity. 
	\end{theorem}
	\begin{proof}
		\underline{Egalitarian rule $f^{\eg}$}: We prove that $f^{\eg}$ 
		satisfies weak line-up enlargement 
		monotonicity.  Let $\pi$ be a winning line-up 
		of the 
		initial election $E=(\mathbf{S},C,P)$
		of 
		egalitarian score $x$. Let $\pi'$ be a winning line-up of the extended 
		election~$E'=(\mathbf{S}',C,P\cup \{p^*\})$ such that there exists a 
		candidate $c\in C$ with
		$c\notin 
		\pi'$ and $c\in \pi$. Let $y$ be the egalitarian score of~$\pi'$. It 
		needs 
		to hold that $x\geq y$, as otherwise $\pi$ cannot be a winning line-up 
		of 
		the initial election. We claim that it is always possible to construct 
		from~$\pi'$ a winning line-up~$\pi''$ of the extended election such 
		that 
		$\pi_{P}\subseteq \pi''_{P\cup \{p^*\}}$. We construct $\pi''$ from 
		$\pi$ and $\pi'$ using 
		the 
		following 
		iterative procedure. Initially, we set $\pi'':=\pi'$. As long as there 
		exists a candidate $c\in C$ with $c\in \pi$ and $c\notin \pi''$, we set 
		$\pi''_{\pi(c)}=c$.  As the egalitarian score of $\pi$ is~$x$, for all 
		$c\in \pi$, it holds 
		that $\score_{\pi(c)}(c)\geq x$. Thus, the 
		egalitarian 
		score of $\pi''$ will not fall below $y$ by the described replacements. 
		Consequently, $\pi''$ is a winning line-up of the extended election 
		$E'$ with all candidates appearing in $\pi$ also appearing in $\pi''$. 
		
		However, $f^{\eg}$ fails to satisfy strong line-up enlargement 
		monotonicity. 
		To see this, consider the following line-up election: 
		\begin{center}
			\begin{tabular}{c|c|c|c|} 
				
				\diagbox{$C$}{$P$}  & $p_1$ & $p_2$ & $p^*$ \\
				\hline
				$a$ & $3$ & $0$ & $0$ \\
				$b$ & $2$ & $3$ & $0$\\
				$c$ & $0$ & $2$ & $0$ \\
				$d$ & $0$ & $0$ & $1$
			\end{tabular}
		\end{center}
		We consider the election on $p_1$ and $p_2$ as the initial election. In 
		this initial election, 
		$(a,b)$ is the unique winning line-up. However, as the new position 
		decreases the lowest score in an optimal outcome, $(b,c,d)$ is also a 
		winning line-up in the extended election. Hence, strong 
		line-up-monotonicity is violated.  
		
		\medskip
		\noindent \underline{Utilitarian rule $f^{\ut}$}: For a proof 
		that$f^{\ut}$ satisfies weak 
		line-up enlargement monotonicity see \Cref{th:utilLEM}. 
		
		Using a similar argument as in \Cref{th:utilLEM}, we prove that 
		$f^{\ut}$
		is even strong line-up enlargement monotone. Let $\pi'$ be a winning 
		outcome
		of the extended election and $\pi$ a winning outcome of the original 
		election
		where there exists a candidate $c\in C$ that is part of $\pi$ but not of
		$\pi'$, that is, $c\in \pi$ but $c\notin \pi'$. To prove that for every
		winning outcome of the extended election, there always exists a winning
		outcome of the original election consisting of a subset of 
		candidates of
		the extended outcome, we construct from~$\pi$ an outcome of the original
		election $\widetilde{\pi}$ with $\widetilde{\pi}|_{P} \subseteq 
		\pi'_{P\cup
		\{p^*\}}$ as follows. We start by setting~$\widetilde{\pi}:=\pi$.
		Subsequently, as long as there exists a position $p\in P$ with
		$\widetilde{\pi}_p\notin \pi'$ or an unassigned position $p\in P$, we 
		set
		$\widetilde{\pi}_p=\pi'_p$. If $\pi'_p$ was already assigned to some 
		other
		position in $\widetilde{\pi}$, then we delete $\pi'_p$ from the 
		position she
		was assigned to in $\widetilde{\pi}$.
		
		It obviously holds that $\widetilde{\pi}\subseteq \pi'$. We claim that
		$\widetilde{\pi}$ is a winning line-up of the initial election. For the
		sake of contradiction, let us assume that this is not the case, that 
		is,~$\Lambda^{\ut}(\widetilde{\pi})<\Lambda^{\ut}(\pi)$. Let 
		$\widetilde{P}\subseteq P$ be the set of positions which have been 
		modified 
		in~$\widetilde{\pi}$, where it needs to hold that 	
		$\Lambda^{\ut}(\pi'|_{\widetilde{P}})<\Lambda^{\ut}(\pi|_{\widetilde{P}})$.
		Let~$\pi'_{\textrm{alt}}$ be a line-up resulting from
		taking~$\pi'$ and modifying it such that
		$\pi'_{\textrm{alt}}|_{\widetilde{P}}=\pi|_{\widetilde{P}}$.
		Consequently, $\pi'_\textrm{alt}$ has a higher utilitarian score
		than that of~$\pi'$. It remains to argue that after the modification 
		$\pi'_\textrm{alt}$ remains a valid line-up, that is, every candidate 
		is 
		assigned at most once in $\pi'_\textrm{alt}$. To prove this it is 
		enough to
		show 
		that, for each $p\in \widetilde{P}$, if $\pi_p\in \pi'$, then 
		$\pi'(\pi_p)\in
		\widetilde{P}$. Let us fix some position~$p \in \widetilde{P}$ such
		that~$\pi_{p} \in \pi'$. Note that position~$p$ had to be empty at some 
		point
		during the construction of $\widetilde{\pi}$. This implies that	in
		some preceding step, $\pi_p$ had been assigned to some other position
		in~$\widetilde{\pi}$. However, by the construction
		of~$\widetilde{\pi}$, each $c\in \pi'$ can be only assigned to
		position~$\pi'(c)$ in $\widetilde{\pi}$. Thus, since $\pi_p$ must have 
		once
		been assigned, $\pi_p$ was assigned to position $\pi'(\pi_p)$; hence,
		$\pi'(\pi_p)\in \widetilde{P}$. So $\pi'_{\textrm{alt}}$ is indeed a 
		valid
		line-up.
		
		\medskip
		\noindent 
		\underline{Harmonic rule $f^{\har}$}: To prove that $f^{\har}$ violates 
		weak line-up enlargement monotonicity, consider 
		the following 
		line-up election:
		\begin{center}
			\begin{tabular}{c|c|c|c|} 
				
				\diagbox{$C$}{$P$}  & $p_1$ & $p_2$ & $p^*$ \\
				\hline
				$a$ & $0.5$ & $0$ & $0$ \\
				$b$ & $3$ & $3.3$ & $0$\\
				$c$ & $0$ & $1$ & $0$ \\
				$d$ & $0$ & $0$ & $4$
			\end{tabular}
		\end{center}  Let the election on $p_1$ and $p_2$ be the initial 
		election. The underlying 
		idea of the counterexample is that the weighting of the scores may 
		change after adding an additional position.  In the initial election, 
		$(a,b)$ is the unique winning line-up as 
		$\Lambda(a,b)=\frac{1}{2}\cdot 0.5+3.3>3+\frac{1}{2}\cdot 
		1=\Lambda(b,c)$. 
		However, in the full election, $(b,c,d)$ is the unique winning 
		line-up, as 
		$\Lambda(b,c,d)=\frac{1}{2}\cdot 3+\frac{1}{3}\cdot 
		1+4>\frac{1}{3}\cdot 0.5+\frac{1}{2}\cdot 3.3+4=\Lambda(a,b,d)$.\\
		
		\medskip
		\noindent \underline{Inverse harmonic rule $f^{\Ihar}$}: To prove that 
		$f^{\Ihar}$ violates weak line-up enlargement monotonicity, let us 
		consider the 
		following line-up election:
		\begin{center}
			\begin{tabular}{c|c|c|c|} 
				
				\diagbox{$C$}{$P$}  & $p_1$ & $p_2$ & $p^*$ \\
				\hline
				$a$ & $1.7$ & $0$ & $0$ \\
				$b$ & $3$ & $1.7$ & $0$\\
				$c$ & $0$ & $1$ & $0$ \\
				$d$ & $0$ & $0$ & $1$
			\end{tabular}
		\end{center}  Again, let the election on $p_1$ and $p_2$ be the initial 
		election. In the initial election, 
		$(a,b)$ is the unique winning line-up as 
		$\Lambda(a,b)=\frac{1}{2}\cdot 1.7+1.7>\frac{1}{2}\cdot 
		3+1=\Lambda(b,c)$. 
		However, in the full election, $(b,c,d)$ is the unique winning 
		line-up, as 
		$\Lambda(b,c,d)=\frac{1}{3}\cdot 3+\frac{1}{2}\cdot 
		1+1>\frac{1}{3}\cdot 1.7+\frac{1}{2}\cdot 1.7+1=\Lambda(a,b,d)$.
	\end{proof}

	\subsubsection{Sequential rules}
	We now turn to sequential rules and prove all statements from 
	\autoref{ta:sum}. We group again the statements for all three rules by the 
	axioms. Note that for the fixed-order rule, without loss of generality, we 
	assume that the positions in $P=\{p_1,\dots, p_q\}$ are filled in 
	increasing 
	order of their index, that is, $p_1$ is filled first, then $p_2$ 
	is filled, and so on.
	\begin{theorem}
		All sequential rules satisfy strong non-wastefulness.
	\end{theorem}
	\begin{proof}
		Assume that $\pi$ is selected as a winning line-up by some sequential 
		rule~$f^{\seq}$ and that $\pi$ is wasteful because there exist a 
		$c\notin \pi$ 
		and a position~$p\in P$ such that $\score_{p}(c)>\score_{p}(\pi)$. This 
		implies 
		that $c$ would have been selected instead of $\pi_p$ to fill position 
		$p$.
	\end{proof}

	\begin{theorem}
		The fixed-order rule, $f^{\seq}_{\fix}$, and the max-first rule, 
		$f^{\seq}_{\ma}$, are weakly but not strongly 
		score Pareto optimal. The min-first rule, $f^{\seq}_{\mi}$, fails weak 
		score Pareto 
		optimality. 
	\end{theorem}
	\begin{proof}
		\underline{Fixed-order rule $f^{\seq}_{\fix}$}: We prove that 
		$f^{\seq}_{\fix}$ is weak score Pareto optimal. 
		Assume that 
		$\pi$ is a 
		winning line-up 
		which is score Pareto dominated by some line-up score Pareto optimal 
		line-up $\pi'$. Then, 
		it 
		follows that there exists a winning outcome $\widetilde{\pi}$ that 
		score 
		Pareto 
		dominates $\pi$ on the first $i$ positions for some $i\in [1,k]$: This 
		can 
		be constructed by choosing at every position $p$ the candidate with the 
		maximum score. If ties are present, then choose the candidate $\pi'_p$ 
		if she 
		is 
		among the candidates with the highest score and otherwise arbitrarily 
		some 
		candidate with the highest score. If, for all positions~$p\in 
		P$,~$\pi'_p$ got selected, then we are done, as this implies that the 
		score 
		Pareto optimal line-up~$\pi'$ is winning under the fixed-order rule. 
		Otherwise, let~$p_i$ be 
		the first position where~$\pi'_{p_i}$ is not 
		selected, 
		which implies that there exists a candidate with a higher score for 
		$p_i$ 
		and 
		that this candidate is selected. Thereby, the score of 
		$\widetilde{\pi}$ is 
		equal to the score of $\pi'$ on the first~$i-1$ positions and higher 
		than 
		the 
		score of $\pi'$ at position $p_i$. Now, either $\widetilde{\pi}$ is 
		score Pareto optimal and we are done or there exists again a line-up 
		that score Pareto dominates $\widetilde{\pi}$. If this is again the case, 
		then we 
		can again construct a winning outcome~$\widetilde{\pi}'$ that 
		score 
		Pareto 
		dominates $\widetilde{\pi}$ on the first $i$ positions for some $i\in 
		[1,k]$. To prove that this
		process converges to an outcome outputted by the fixed-order rule that 
		is score Pareto optimal, we read 
		the 
		score vector of an outcome $(\score_{p_1}(\pi),\dots, 
		\score_{p_k}(\pi))$ 
		as a 
		decimal number and call this the line-up's score number. As described 
		above, 
		each time a 
		winning outcome~$\pi$ is score Pareto dominated by some outcome 
		$\pi'$, 
		there exists a winning line-up~$\widetilde{\pi}$ with a higher score 
		number than 
		$\pi$. The winning line-up with the highest score number 
		needs to be 
		score
		Pareto optimal.
		
		In contrast to this, $f^{\seq}_{\fix}$ is not strongly score Pareto 
		optimal. In the following line-up election, $(a,b)$ is a winning 
		line-up, where $(a,b)$ is score Pareto dominated by $(b,a)$, which, of 
		course, also is a winning outcome.  
		\begin{center}
			\begin{tabular}{c|c|c|} 
				
				\diagbox{$C$}{$P$}  & $p_1$ & $p_2$ \\
				\hline
				$a$ & $2$ & $1 $ \\
				$b$ & $2$ & $0$
			\end{tabular}
		\end{center}

		\medskip 
		\noindent \underline{Max-first rule $f^{\seq}_{\ma}$}: To 
		prove that $f^{\seq}_{\ma}$ satisfies weak score 
		Pareto optimality, we apply the 
		same reasoning 
		as before. However, in this case, the score number of an 
		outcome is 
		the score vector sorted in decreasing order. The example from above can 
		also be used to prove that $f^{\seq}_{\ma}$ does not satisfy strong 
		score Pareto optimality. 
		
		\medskip
		\noindent \underline{Min-first rule $f^{\seq}_{\mi}$}: Finally, we 
		prove that $f^{\seq}_{\mi}$ violates weak score Pareto optimality. 
		Here, it is not 
		possible to 
		modify a 
		dominated 
		outcome such that certain positions or the maximum scores always 
		improve. 
		To prove that $f^{\seq}_{\mi}$ violates weak score Pareto optimality, 
		consider the following line-up election:
		\begin{center}
			\begin{tabular}{c|c|c|c} 
				\diagbox{$C$}{$P$} & $p_1$ & $p_2$ & $p_3$  \\
				\hline
				$a$ & 1 & 1  & 3\\
				$b$ & 1 & 3 & 2\\
				$c$ & 0 & 0  & 0
			\end{tabular}
		\end{center}
		Because all candidates have at most score one on $p_1$, $p_1$ is 
		required to pick first but is indifferent between $a$ and $b$. 
		Assuming that $p_1$ picks $a$, $p_3$ is required to pick next and picks 
		$b$. Then, $p_2$ gets assigned $c$. Assuming that $p_1$ picks $b$, 
		$p_2$ is 
		required 
		to pick 
		next and picks $a$. Then, $p_3$ gets assigned $c$. Consequently, the 
		two 
		outcomes of 
		the game are~$(a,c,b)$ and $(b,a,c)$. 
		However $(a,c,b)$ is score Pareto dominated by $(b,c,a)$ and 
		$(b,a,c)$ 
		is score Pareto dominated by $(a,b,c)$. Note that $(a,b,c)$ is not a 
		winning line-up, as in the case where position~$p_1$ picks $a$, 
		position $p_3$ is required to select next, and not, as it would be 
		necessary to achieve this outcome, position $p_2$. Thereby, both 
		winning 
		line-ups 
		returned by $f^{\seq}_{\min}$ are score Pareto dominated.
	\end{proof}

		\begin{theorem}
		The max-first rule, $f^{\seq}_{\ma}$, satisfies strong reasonable 
		satisfaction.  
		The fixed-order rule, $f^{\seq}_{\fix}$, and the min-first rule, 
		$f^{\seq}_{\mi}$, violate weak reasonable 
		satisfaction.
	\end{theorem}
	\begin{proof}
		\underline{Fixed-order rule $f^{\seq}_{\fix}$ and min-first rule 
		$f^{\seq}_{\mi}$}: We start by presenting a counterexample for 
		$f^{\seq}_{\fix}$  and $f^{\seq}_{\mi}$ 
		proving that both rules do not satisfy weak reasonable satisfaction.
		\begin{center}
			\begin{tabular}{c|c|c|} 
				
				\diagbox{$C$}{$P$} & $p_1$ & $p_2$ \\
				\hline
				$a$ & $1$ & $2$ \\
				$b$ & $0$ & $1$
			\end{tabular}
		\end{center}
		Both $f^{\seq}_{\fix}$  and $f^{\seq}_{\mi}$ return $(a,b)$ as the 
		unique 
		winner. However, in the outcome~$(b,a)$, the candidate-position pair 
		$(a,p_2)$ is reasonably dissatisfied.
		
		\medskip 
		\noindent For the \underline{max-first rule $f^{\seq}_{\ma}$}, see 
		\Cref{th::rSat}.
	\end{proof}

	\begin{theorem}
		The fixed-order rule, $f^{\seq}_{\fix}$, satisfies weak but not strong 
		score consistency. 
		The max-first rule, $f^{\seq}_{\ma}$, and the min-first rule, 
		$f^{\seq}_{\mi}$, both violate weak 
		score consistency.
	\end{theorem}
	\begin{proof}
		\underline{Fixed-order rule $f^{\seq}_{\fix}$}:
		If $\pi$ 
		is a winning line-up in both 
		$E=(\mathbf{S},C,P)$ and~$E'=(\mathbf{S}',C,P)$, then
		it holds that $\pi_{p_{i}}$ had the maximum score for position $p$ 
		among all candidates not already assigned to some position $\{p_j\mid 
		1\leq
		j<i\}$ 
		in both $E$ and $E'$. By induction and by the fact that the score of a 
		candidate for a position in $(\mathbf{S}+\mathbf{S}',C,P)$ is the sum 
		of her scores in $E$ and $E'$, it follows that 
		the same is true in $(\mathbf{S}+\mathbf{S}',C,P)$. 
		
		However, the strong version of this property does not hold: 
		\begin{center}
			$E$: \begin{tabular}{c|c|c} 
				\diagbox{$C$}{$P$} & $p_1$ & $p_2$ \\
				\hline
				$a$ & 3 & 1 \\
				$b$ & 3 & 3 \\
				$c$ & 0 & 3 
			\end{tabular}
		\qquad
			$E'$:
			\begin{tabular}{c|c|c} 
				\diagbox{$C$}{$P$} & $p_1$ & $p_2$ \\
				\hline
				$a$ & 3 & 3 \\
				$b$ & 3 & 3 \\
				$c$ & 0 & 1
			\end{tabular}
		\qquad
			$E^*$: 
			\begin{tabular}{c|c|c} 
				\diagbox{$C$}{$P$} & $p_1$ & $p_2$ \\
				\hline
				$a$ & 6 & 4 \\
				$b$ & 6 & 6 \\
				$c$ & 0 & 4
			\end{tabular}
		\end{center} In both elections~$E$ and $E'$, $(a,b)$ is a winning line-up. 
		Strong score 
		consistency is violated, as for the combined election $E^*$, $(b,a)$ is 
		also 
		a 
		winning 
		line-up. However, $(b,a)$ is not a winning line-up in election~$E$.
		
		\medskip
		\noindent \underline{Max-first rule $f^{\seq}_{\ma}$}: Weak score 
		consistency does not hold 
		by the following example: 
		\begin{center}
			$E$: \begin{tabular}{c|c|c} 
				\diagbox{$C$}{$P$} & $p_1$ & $p_2$ \\
				\hline
				$a$ & 4 & 3 \\
				$b$ & 2 & 1 
			\end{tabular}
			\qquad
			$E'$:
			\begin{tabular}{c|c|c} 
				\diagbox{$C$}{$P$} & $p_1$ & $p_2$ \\
				\hline
				$a$ & 1 & 3 \\
				$b$ & 2 & 4
			\end{tabular}
			\qquad
			$E^*$: 
			\begin{tabular}{c|c|c} 
				\diagbox{$C$}{$P$} & $p_1$ & $p_2$ \\
				\hline
				$a$ & 5 & 6 \\
				$b$ & 4 & 5
			\end{tabular}
		\end{center}
	In both elections~$E$ and $E'$, $(a,b)$ is the unique winning line-up, 
	while in the combined election $E^*$, $(b,a)$ is the unique winning 
	line-up. 
		
		\medskip 
		\noindent \underline{Min-first rule $f^{\seq}_{\mi}$}: The min-first 
		rule violates weak score consistency by 
		examining the following line-up election:
		\begin{center}
			$E$: \begin{tabular}{c|c|c} 
				\diagbox{$C$}{$P$} & $p_1$ & $p_2$ \\
				\hline
				$a$ & 1 & 3 \\
				$b$ & 0 & 0 
			\end{tabular}
			\qquad
			$E'$:
			\begin{tabular}{c|c|c} 
				\diagbox{$C$}{$P$} & $p_1$ & $p_2$ \\
				\hline
				$a$ & 0 & 0 \\
				$b$ & 4 & 1 
			\end{tabular}
			\qquad
			$E^*$: 
			\begin{tabular}{c|c|c} 
				\diagbox{$C$}{$P$} & $p_1$ & $p_2$ \\
				\hline
				$a$ & 1 & 3 \\
				$b$ & 4 & 1 
			\end{tabular}
		\end{center} 
		In both elections~$E$ and $E'$, $(a,b)$ is selected as the unique 
		winning line-up.
		However, in the combined election $E^*$, the unique winning  
		line-up is $(b,a)$.
	\end{proof}

	\begin{theorem}
		All three sequential rules satisfy weak position consistency but fail 
		to satisfy 
		strong position consistency.
	\end{theorem}
	\begin{proof}
		We start by proving that all considered sequential rules satisfy weak 
		position consistency, before presenting a line-up election on which all 
		all rules violate strong position consistency.
		
		\medskip
		\noindent \underline{Fixed-order rule $f^{\seq}_{\fix}$}: Let 
		$E=(\mathbf{S},C,P)$ be a 
		line-up election and $P',P''\subseteq P$ two subsets of positions such 
		that 
		$P'\cup P''=P$. To prove that weak position 
		consistency is satisfied, let $\pi'$ be a 
		winning line-up in the election $E'=(\mathbf{S},C,P')$ on $P'$ and 
		$\pi''$ be an 
		overlapping-disjoint 
		winning line-up of the election $E''=(\mathbf{S},C,P'')$ on $P''$. If 
		we now consider a 
		combination of 
		the two elections where the order within the sets $P'$ and $P''$ 
		remains 
		unchanged, then $\pi'\cup \pi''$ is a winning line-up. This follows 
		from the 
		fact 
		that $\pi'$ and $\pi''$ are disjoint on the disjoint sets of positions. 
		Thus, 
		no position will ``steal'' a candidate from some other position.
		
		\medskip \noindent \underline{Max-first rule $f^{\seq}_{\ma}$ and  
		min-first rule 
		$f^{\seq}_{\mi}$}: For the max-first rule and min-first rule, the 
		argument is 
		analgous to the above one, as again the order within $P$ in which the 
		positions 
		are 
		filled 
		and the order within $P'$ in which the positions are filled remain 
		unchanged. 
		
		\medskip \noindent
		To prove that none of the three rules satisfies strong position 
		consistency, consider the following line-up election 
		$E=(\mathbf{S},C,P)$:
		\begin{center}
			\begin{tabular}{c|c|c|c} 
				
				\diagbox{$C$}{$P$} & $p_1$ & $p_2$ & $p_3$\\
				\hline
				$a$ & $1$ & $0$ & $0$ \\
				$b$ & $0$ & $1$ & $0$ \\
				$c$ & $1$ & $0$ & $1$ \\
			\end{tabular}
		\end{center}
		Let $P'=\{p_1,p_2\}$ and $P''=\{p_2,p_3\}$. The winning line-ups of the 
		election $E'=(\mathbf{S},C,P')$ are $(a,b)$ and $(c,b)$. The unique 
		winning line-up of 
		the election $E''=(\mathbf{S},C,P'')$ is $(b,c)$. The two line-ups 
		$(a,b)$ and $(b,c)$ 
		are overlapping-disjoint 
		coalitions. 
		Nevertheless, for the winning line-up $\pi^*=(c,b,a)$ of the full 
		election $E$, 
		there do not exist two overlapping-disjoint winning coalitions $\pi'$ 
		and $\pi''$ 
		such 
		that $\pi^*=\pi\cup \pi'$.
	\end{proof}
	
	\begin{theorem}
		The fixed-order rule, $f^{\seq}_{\fix}$, and the max-first rule, 
		$f^{\seq}_{\ma}$, both satisfy strong 
		monotonicity. The min-first rule, $f^{\seq}_{\mi}$, violates weak 
		monotonicity.
	\end{theorem}
	\begin{proof}
		\underline{Fixed-order rule $f^{\seq}_{\fix}$}: The fixed-order rule 
		satisfies strong monotonicity, as if the score of 
		some 
		candidate gets 
		increased for the position she got
		assigned to, then this position will select again this candidate once 
		it is 
		the 
		position's turn.
		
		\medskip 
		\noindent \underline{Max-first rule $f^{\seq}_{\ma}$}: If the score of 
		an assigned 
		candidate-position 
		pair from a line-up $\pi$
		is increased, then this pair moves up in the list of candidate-position 
		pairs sorted 
		by scores. In the max-first rule, the first still realizable element 
		from this list is always chosen as the next 
		assigned 
		pair. Thereby, breaking ties in the right way, all candidate-position 
		pairs from $\pi$ will still be selected. Moreover, no new winning 
		line-ups will be created.
		
		\medskip
		\noindent \underline{Min-first rule $f^{\seq}_{\mi}$}: To prove 
		that $f^{\seq}_{\mi}$ violates weak monotonicity, 
		consider the 
		following 
		election: 
		\begin{center}
			\begin{tabular}{c|c|c} 
				\diagbox{$C$}{$P$} & $p_1$ & $p_2$ \\
				\hline
				$a$ & 2 & 1 \\
				$b$ & 0 & 0 
			\end{tabular}
		\end{center} 
		The unique winning line-up is $(b,a)$. However, after increasing 
		$\score_{p_2}(a)$ by two, the unique winning line-up is $(a,b)$.
	\end{proof}

	\begin{theorem} \label{th:LEMM}
		The fixed-order rule, $f^{\seq}_{\fix}$, and the max-first rule, 
		$f^{\seq}_{\ma}$, both satisfies strong 
		line-up enlargement monotonicity. The min-first rule, $f^{\seq}_{\mi}$, 
		violates weak line-up enlargement monotonicity.  
	\end{theorem}
	\begin{proof}
		\underline{Fixed-order rule $f^{\seq}_{\fix}$}: To prove that  
		$f^{\seq}_{\fix}$ satisfies weak line-up enlargement monotonicity, let 
		$P=\{p_1,\dots, 
		p_q\}$ and let 
		$j$ 
		be the place in the ordering of positions at which~$p^*$ is inserted. 
		To prove weak line-up enlargement monotonicity, let $\pi$ be a winning 
		outcome of the initial election without $p^*$. From 
		this, we 
		construct an outcome $\pi'$ of the extended election as follows. 
		For $i\in [1,j-1]$, we assign $\pi_{p_i}$ to $p_i$. For $i=j$, we 
		assign to~$p^*$ the remaining candidate with the highest score on 
		$p^*$. For 
		$i\in 
		[j+1,q+1]$, we assign~$\pi_{p_{i-1}}$ to~$p_{i-1}$  if~$\pi_{p_{i-1}}$ 
		is still unassigned 
		and 
		has the highest score 
		among all unassigned candidates; otherwise, we pick an arbitrary 
		unassigned 
		candidate with the highest score on~$p_{j-1}$. The line-up $\pi'$ is 
		obviously a 
		winning line-up under the sequential rule in the extended election.
		
		We prove by induction that for each step $i\in [1,q+1]$ of the 
		construction, the candidates who have been assigned to $\pi'$ at steps 
		$1,\dots i$ are a superset of 
		$\pi_{\{p_1,\dots,p_{i-1}\}}$. From this it immediately follows that 
		$\pi_{P}\subseteq \pi'_{P\cup \{p^*\}}$. For $i\in [1,j]$, the 
		induction statement is 
		trivially 
		fulfilled, as it holds that $\pi'_{p_i}=\pi_{p_i},$ for all $i\in 
		[1,j-1]$. 
		Let  
		$i\in [j+1,q+1]$ and assume that the hypothesis holds for $i-1$. 
		At step $i$, we fill position $p_{i-1}$. 
		If 
		it holds at this step that $\pi_{p_{i-1}}$ has been already assigned 
		earlier, then we are 
		done. If this it not the case, then we claim that $\pi_{p_{i-1}}$ has 
		the 
		highest score on $p_{i-1}$ among all remaining candidates and is 
		therefore 
		assigned to $p_{i-1}$ in $\pi'$. Recalling that 
		by the induction hypothesis the remaining candidates are a subset of 
		$C\setminus \{\pi_{p_1},\dots, \pi_{p_{i-2}}\}$ and that $\pi$ is a 
		winning 
		line-up of the initial election, that is, $\pi_{p_{i-1}}$ has the 
		highest score among all 
		$C\setminus 
		\{\pi_{p_1},\dots, \pi_{p_{i-2}}\}$, the claim holds. From this, the 
		induction step directly follows. 
		
		It is possible to prove that the sequential rule also satisfies 
		condition (b) of line-up enlargement monotonicity and, thereby, also 
		strong line-up enlargement monotonicity using a similar argument. Let 
		$\pi'$ be a 
		winning line-up of the extended election. We construct a winning 
		line-up 
		$\pi$ of the initial election as follows. For $i\in [1,j-1]$, we 
		set~$\pi_{p_i}=\pi'_{p_i}$. For $i\in [j,q]$, we assign to $\pi_{p_i}$ 
		the 
		candidate with the highest score among the remaining candidates and 
		break 
		ties by always
		selecting an unassigned candidate from~$\{\pi'_{p_j},\dots 
		\pi'_{p_{i}},\pi'_{p^*} 
		\}$ if possible. Clearly, $\pi$ is a winning line-up of the initial 
		election. For the sake of contradiction, let $i\in [j,q]$ be the 
		smallest 
		index 
		where a candidate $c\in C$ not being part of $\{\pi'_{p_j},\dots 
		\pi'_{p_i},\pi'_{p^*} \}$ is selected. However, this implies that 
		$\pi'$ 
		is 
		not a winning line-up of the extended election, as $c$ has a higher 
		score than $\pi'_{p_i}$ on 
		$p_i$. Note further that~$c$ is by definition also still unassigned 
		when 
		$\pi_{p_i}$ is filled in $\pi'$. Hence, such an index cannot exist 
		proving 
		that $\pi_{P}\subseteq \pi'_{P\cup \{p^*\}}$. 
		
		\medskip 
		\noindent \underline{Max-first rule $f^{\seq}_{\ma}$}: To prove that 
		$f^{\seq}_{\ma}$ satisfies weak line-up enlargement monotonicity, 
		let $\pi$ be a winning 
		outcome of 
		some election $(\mathbf{S},C,P)$. From $\pi$, we construct an 
		outcome~$\pi'$ of the 
		extended 
		election $(\mathbf{S}',C,P\cup \{p_{q+1}\})$ as follows. At each step 
		$i\in 
		[1,q+1]$, we assign the 
		candidate-position 
		pair $(c,p)\in C\times P$ with the highest score among all remaining 
		pairs. If 
		there 
		exist multiple such pairs, then we select a pair with $(c,p)\in \pi$ 
		and if multiple such pairs exist, then we select the pair that has been 
		first assigned 
		in the construction of $\pi$. Line-up
		$\pi'$ is 
		obviously a winning line-up under the max-first rule. It remains to 
		prove that 
		$\pi_{P}\subseteq \pi'_{P\cup \{p_{q+1}\}}$. 
		
		Let $j^*\in [1,q]$ denote the step in which $p_{q+1}$ is filled. 
		Moreover, let $i_1,\dots i_q$ be a list of indices such that position 
		$p_{i_j}$ 
		is filled in the $j$-th step when creating~$\pi$ and let~$i'_1,\dots 
		i'_{q+1}$ be a 
		list of indices such that position $p_{i'_j}$ is filled in the
		$j$-th 
		step when 
		creating $\pi'$. We prove by induction that for all $j\in [1,q+1]$ it 
		holds 
		that the 
		candidates assigned to $\pi'$ after~$j$ steps are a 
		superset of 
		the candidates
		assigned to $\pi$ after $j-1$ steps, which directly implies that 
		$\pi_{P}\subseteq \pi'_{P\cup \{p_{q+1}\}}$. For $j\in [1,j^*]$, the 
		hypothesis clearly holds, as $i_k=i'_k$ 
		and~$\pi_{p_{i_k}}=\pi'_{p_{i'_k}}, \forall k\in 
		[1,j^*-1]$. For $j\in [j^*+1,q+1]$, we assume that the induction 
		statement 
		holds 
		for $j-1$. At step $j$, if $\pi_{p_{i_{j-1}}}$ is already assigned in 
		$\pi'$ before 
		step $j$, then
		we are done. Otherwise, we claim that the candidate-position pair 
		$(p_{i_{j-1}},\pi_{p_{i_{j-1}}})$ will be assigned in this step in the 
		construction of $\pi'$. First 
		of all, 
		note that by the induction hypothesis and as this pair has been 
		assigned in $\pi$, $\pi_{p_{i_{j-1}}}$ has the 
		highest 
		score among the remaining candidates on $p_{i_{j-1}}$. Moreover, note 
		that by 
		the induction hypothesis, there cannot exist an earlier step where 
		$p_{i_{j-1}}$ was assigned some other candidate because this requires 
		that this candidate has a higher score on $p_{i_{j-1}}$ than 
		$\pi_{p_{i_{j-1}}}$ and thereby that this candidate 
		would have 
		 also been assigned to $p_{i_{j-1}}$ in the construction of 
		$\pi$. 
		Furthermore, 
		it needs to hold that $(p_{i_{j-1}},\pi_{p_{i_{j-1}}})$ is the 
		candidate-position pair with the highest score, as there only exist 
		$j-2$ pairs
		with a higher score in the construction of $\pi$ 
		and 
		thereby at most~$j-1$ pairs with a higher score that were all already 
		chosen before in the 
		construction of~$\pi'$.	
		
		To prove that condition (b) of line-up-enlargement-monotonicity holds 
		as well, let $\pi'$ be a 
		winning outcome of the larger 
		election and 
		let $\pi$ be an outcome of the initial election on $P$ constructed from 
		$\pi'$ 
		as follows. At each 
		step $i\in 
		[1,q]$, we assign the candidate-position 
		pair $(c,p)\in C\times P$ with the highest score among all remaining 
		pairs. If 
		there 
		exist multiple such pairs, then we select a pair $(c,p)\in \pi'$, and 
		if multiple such pairs exist, then we assign the pair that has been 
		first assigned 
		in the construction of $\pi'$. 
		
		For the sake of contradiction, let $j$ be the first step in the 
		construction of~$\pi$ in which a candidate-position pair $(c,p)\in 
		C\times P$ with $c\notin \pi'$ is assigned. As otherwise $c$ would have 
		been assigned to $p$ in $\pi'$, this implies that $\pi'_p$ has already 
		been assigned at step $j'<j$ in the construction of $\pi$. Thereby, it 
		needs to hold that $\score_{\pi(c)}(c)>\score_{p}(c)$. This implies 
		that in $\pi'$ it needs to hold that $\score_{\pi(c)}(\pi')\geq 
		\score_{\pi(c)}(c)$. However, this in turn implies that $\pi'_{\pi(c)}$ 
		must have been assigned in $\pi$ at some step $j''<j'$ and that it 
		needs to hold that 
		$\score_{\pi(\pi'_{\pi(c)})}(\pi'_{\pi(c)})>\score_{\pi(c)}(\pi'_{\pi(c)})$.
		 Then, again the argument from above applies eventually resulting in a 
		contradiction once step $1$ is reached.

		\medskip
		\noindent \underline{Min-first rule $f^{\seq}_{\mi}$}: Consider the 
		following election where we denote 
		by 
		$(\mathbf{S},C,P)$ 
		the election on the first four positions and by $(\mathbf{S}',C,P')$ 
		the 
		election on 
		all five positions. 
		\begin{center}
			\begin{tabular}{c|c|c|c|c|c} 
				\diagbox{$C$}{$P$} & $p_1$ & $p_2$ & $p_3$ & $p_4$ & $p^*$ \\
				\hline
				$a$ & $0$ & $6$ & $4$ & $0$ & $0$\\
				$b$ & $2$ & $3$ & $0$ & $0$ & $0$\\
				$c$ & $7$ & $0$ & $0$ & $5$ & $1$\\
				$d$ & $1$ & $0$ & $0$ & $0$ & $0$\\
				$e$ & $0$ & $0$ & $0$ & $1$ & $0$\\
				$f$ & $0$ & $1$ & $0$ & $0$ & $0$
			\end{tabular}
		\end{center}
		We claim that $f^{\seq}_{\min}(\mathbf{S},C,P)=\{(d,b,a,c)\}$ and that 
		$f^{\seq}_{\min}(\mathbf{S}',C,P')=\{(b,f,a,e,c)\}$, thus violating 
		line-up enlargement 
		monotonicity. For the election on the first four positions,~$p_3$ is 
		required to pick first 
		and picks 
		$a$. After that, $p_2$ is required to pick and picks $b$. Subsequently, 
		$p_4$ 
		picks $c$ and $p_1$ picks $d$. Thus, $(d,b,a,c)$ is the unique winning 
		line-up. For the election on all five positions, 
		the order in 
		which the 
		positions pick will change because of the presence of $p^*$: $p^*$ is 
		required 
		to pick first and picks $c$. Subsequently, $p_4$ picks $e$, and $p_1$ 
		picks $b$. 
		$p_3$ picks $a$, and $p_2$ picks $f$. Hence, $(b,f,a,e,c)$ is the 
		unique winning line-up.
	\end{proof}

\newpage
	\section{Experiments} \label{se:boxplots}
	In the following section, we provide several additional details for the 
	conducted experiments. First of all, the selected soccer formation for our 
	FIFA experiments (\autoref{se:formation}), the definition of the Gini 
	coefficient (\autoref{se:Gini}), and the description of the first model to 
	generate synthetic data (\autoref{se:M1}) are presented. Subsequently, in 
	\autoref{se:CompModels} and \autoref{fi:missres}, we perform a comparison 
	of the different data models and provide the missing diagrams for the 
	analysis of ten-candidates ten-positions elections presented in the main 
	body. Finally, we 
	conclude in 
	\autoref{se:moreResults} with visualizations of our results for two 
	different setups varying the number of candidates and positions.
	\subsection{Soccer Formation} \label{se:formation}
		\autoref{fi:SocLine} contains a visualization of the chosen formation 
		that 
	we want to fill in our experiments based on FIFA data. That is, the 
	positions in the conducted experiments are the different names displayed on 
	the field.
	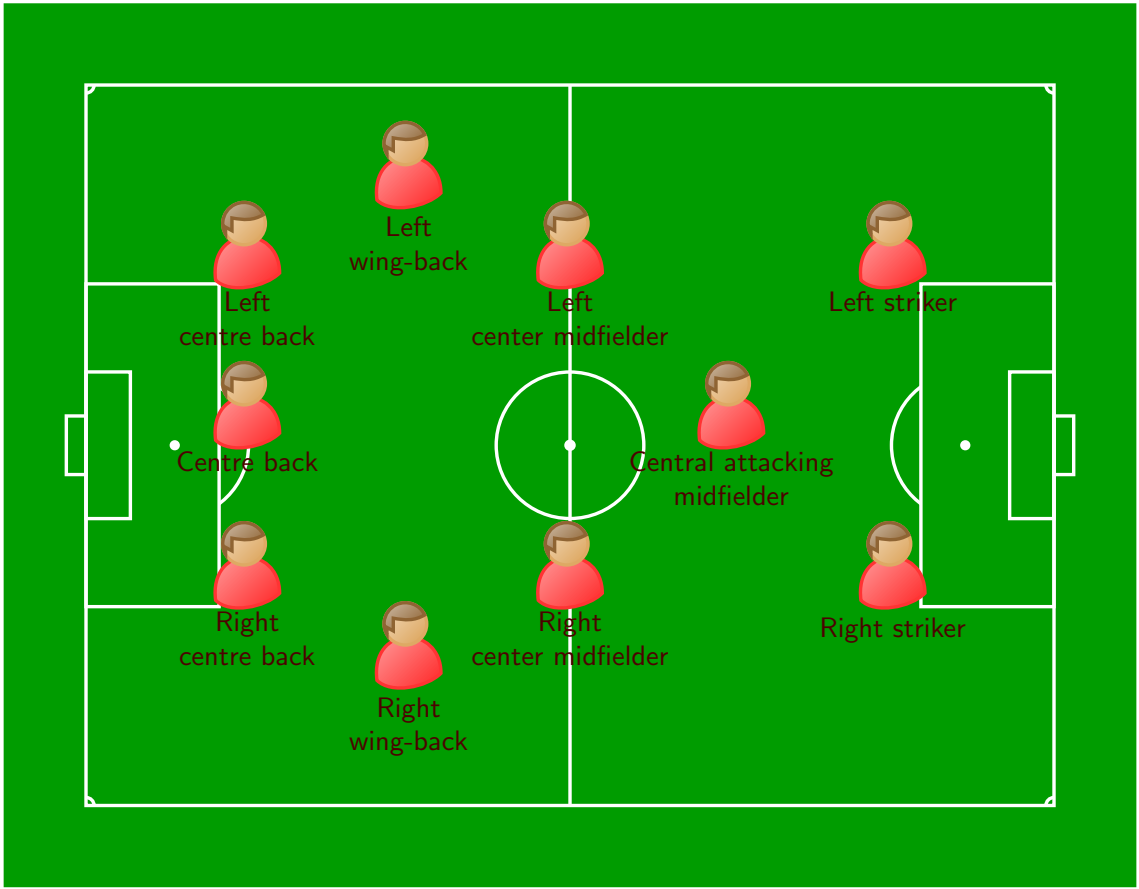
\begin{figure}[!htb]
	\resizebox{\textwidth}{!}{%
	\begin{tikzpicture}
	\fill[field] (-1,-1) rectangle (13,10);
	\node[minimum width=12cm, minimum height=9cm] (contour) at (6,4.5) {};
	
	\draw (contour.north) -- (contour.south);
	\draw (contour.center) circle (0.915cm);
	\fill[white] (contour.center) circle (.5mm);
	
	\area{contour.west}{0}
	\area{contour.east}{180}
	
	\foreach \corner [count=\xi starting from 0] in {south west, south east, 
	north 
		east, north west}{
		\begin{scope}[rotate around={90*\xi:(contour.\corner)}]
		\draw ([xshift=1mm]contour.\corner) arc (0:90:1mm);
		\end{scope}
	}

	\begin{scope}[font=\small\sffamily, text=black!70!red]
	\node[person, minimum size=.8cm, anchor=center, shirt=red] at 
	(10,3) {Right striker};
	
	\node[person, minimum size=.8cm, anchor=center, shirt=red] at 
	(10,7) {Left striker};

	\node[person, minimum size=.8cm, anchor=center, shirt=red, align=center] at 
	(8,5) {Central attacking \\ midfielder};
	\node[person, minimum size=.8cm, anchor=center, shirt=red, align=center] at 
	(6,7) {Left\\ center midfielder };
	\node[person, minimum size=.8cm, anchor=center, shirt=red, align=center] at 
	(6,3) {Right\\ center midfielder };
	
	\node[person, minimum size=.8cm, anchor=center, shirt=red, align=center] at 
	(4,8) {Left\\ wing-back};
	\node[person, minimum size=.8cm, anchor=center, shirt=red, align=center] at 
	(4,2) {Right\\ wing-back};
	
	\node[person, minimum size=.8cm, anchor=center, shirt=red, align=center] at 
	(2,3) {Right\\ centre back};
	\node[person, minimum size=.8cm, anchor=center, shirt=red] at 
	(2,5) {Centre back};
	\node[person, minimum size=.8cm, anchor=center, shirt=red, align=center] at 
	(2,7) {Left\\ centre back};
	
	\end{scope}
	
	\end{tikzpicture}} \caption{Soccer formation used in the experiments on 
	FIFA 
	data.} \label{fi:SocLine}
	\end{figure}

	\subsection{Gini Coefficient} \label{se:Gini}
	
	The Gini coefficient is a metric to 
	measure the dispersion of a probability distribution; it is zero for 
	uniform distributions and one for distributions with a unit step 
	cumulative distribution function. The Gini coefficient of a tuple 
	$X=(x_1,\dots,x_n)$ is defined as: 
	$$\text{gini}(X)=\frac{\sum_{i=1}^{n}\sum_{j=1}^{n} 
	|x_i-x_j|}{2n\sum_{i=1}^{n} x_i}.$$
	
	\subsection{Model M1} \label{se:M1}
	To model that some 
	positions are likely to be correlated, we start by partitioning the 
	positions into three subsets. For each candidate and each subset, we 
	draw a ground qualification $\mu$ uniformly between 0.4 and 0.7. 
	Subsequently, we sample the scores of this player for all positions in this 
	subset 
	independently from a Gaussian distribution with mean 
	$\mu$ 
	and standard deviation $0.15$. The results of the different voting methods 
	on data generated using the $M1$ data can be found in \autoref{fig:M1}. The 
	performance of all voting rules 
	on this dataset is slightly worse then on the FIFA data. However, the 
	differences between the voting rules with respect to the considered metrics 
	is pretty similar to the patterns observed in the FIFA data.  
	
	\subsection{Comparison of Data Models} \label{se:CompModels}
	To compare the three models, as the relative conflicting score, we 
	calculated the difference between the summed score of an utopic outcome 
	and 
	the summed 
	score of a 
	utilitarian outcome divided by the summed score of the 
	utopic outcome. Moreover, we computed for all positions the summed 
	score of all candidates on this position. Subsequently, we calculated the 
	Gini coefficient  of the 
	summed scores of positions 
	to measure the 
	distribution of scores over positions. Lastly, we calculated 
	the social conflict defined as the lowest score in an 
	egalitarian 
	line-up minus the lowest score in an utilitarian line-up. We compare the 
	three data models to each
	other in \autoref{fig:ovInst}.
	On the one hand, the M2 model produces instances with a higher relative 
	conflicting score, a larger Gini coefficient of position scores and a 
	higher 
	social trade-off compared to the other models. On the other hand, the M1 
	model and the FIFA data
	produce similar instances in terms of the considered metrics.
	\begin{figure}[!htb]
	\includegraphics[width=\textwidth]{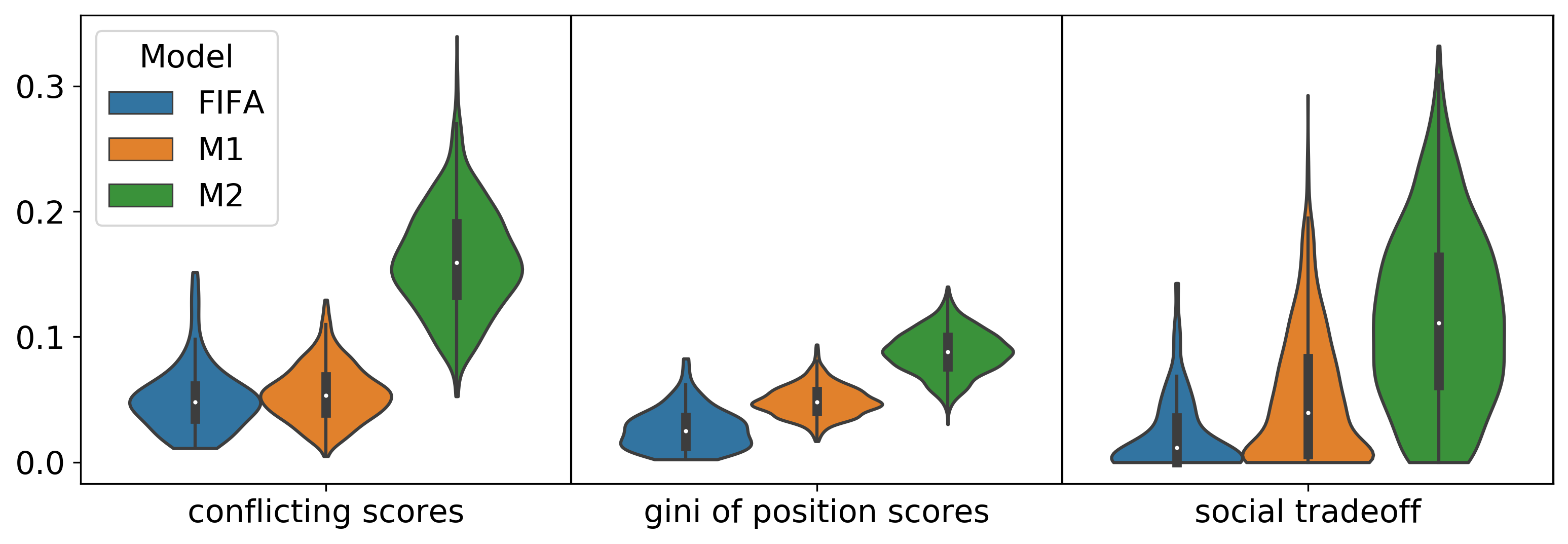}
	\caption{Comparison of different models to generate ten-candidates 
	ten-positions elections.}
	\label{fig:ovInst}
	\end{figure}

	\subsection{Missing Experimental Results for M1 Model} \label{fi:missres} 
	\label{se::M1res}
	In \autoref{fig:M1}, we present a comparison of the different voting rules 
	on M1 and FIFA data for ten players on ten positions. Note that while 
	the general trends are similar on both data sets, all voting rules produce 
	slightly worse results on the M1 data. Moreover, the difference between 
	the voting rules is more profound on the M1 data. 
	
	\begin{figure}[htb!]
		\captionsetup[subfigure]{justification=centering}
		\centering
		\begin{subfigure}[b]{\textwidth}
			\centering
			\includegraphics[width=\textwidth]{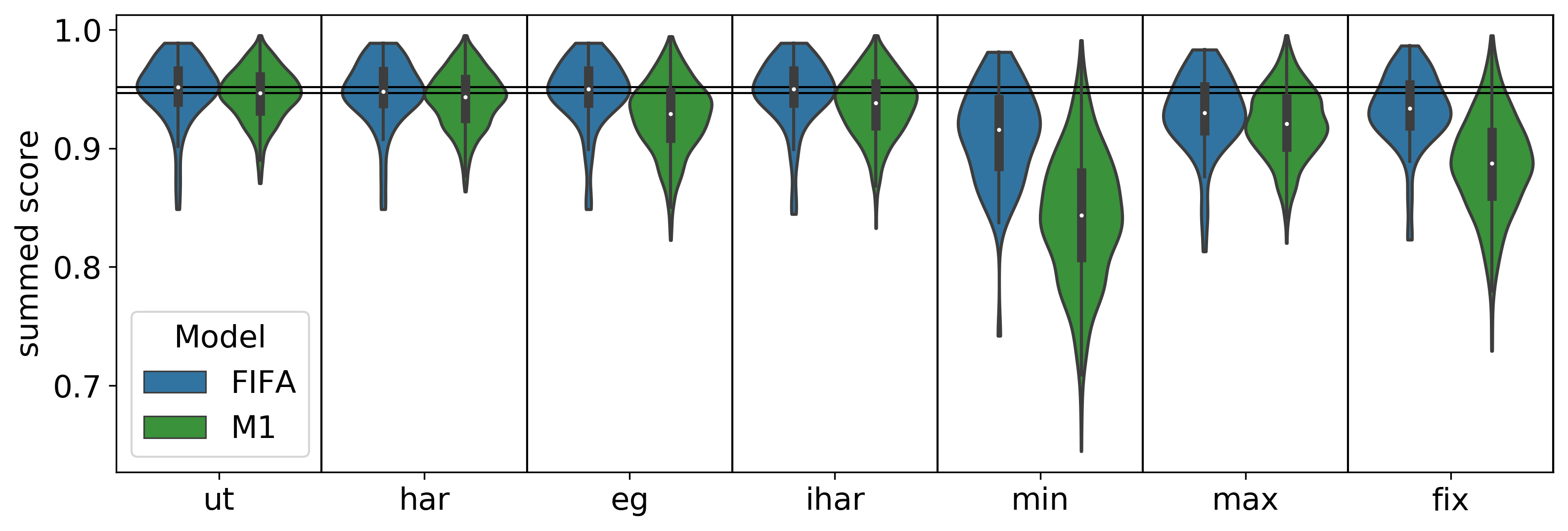}
			\caption{Summed score.} 
			
		\end{subfigure}\vspace*{-0.5cm}
	\end{figure}
	\begin{figure}[htb!]\ContinuedFloat
		\begin{subfigure}[b]{\textwidth}
			\centering
				\includegraphics[width=\textwidth]{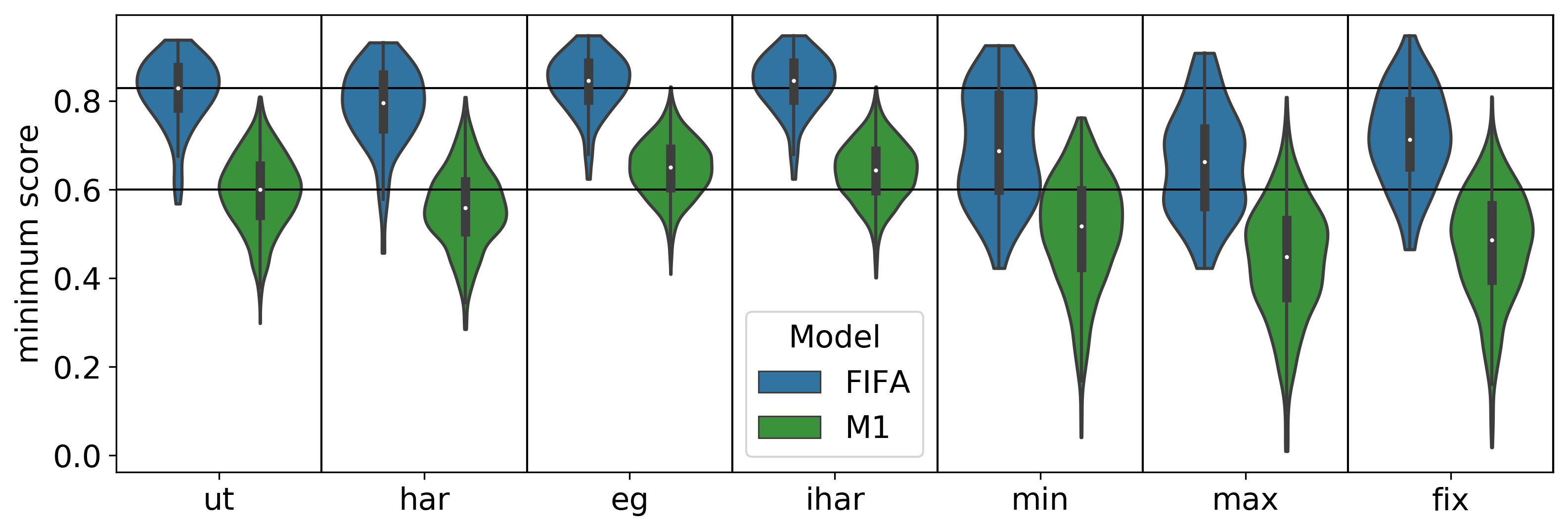}
			\caption{Minimum score.} 
			
		\end{subfigure}\vspace*{-0.5cm}
	\end{figure}
	\begin{figure}[htb!]\ContinuedFloat
		\begin{subfigure}[b]{\textwidth}
			\centering
			\includegraphics[width=\textwidth]{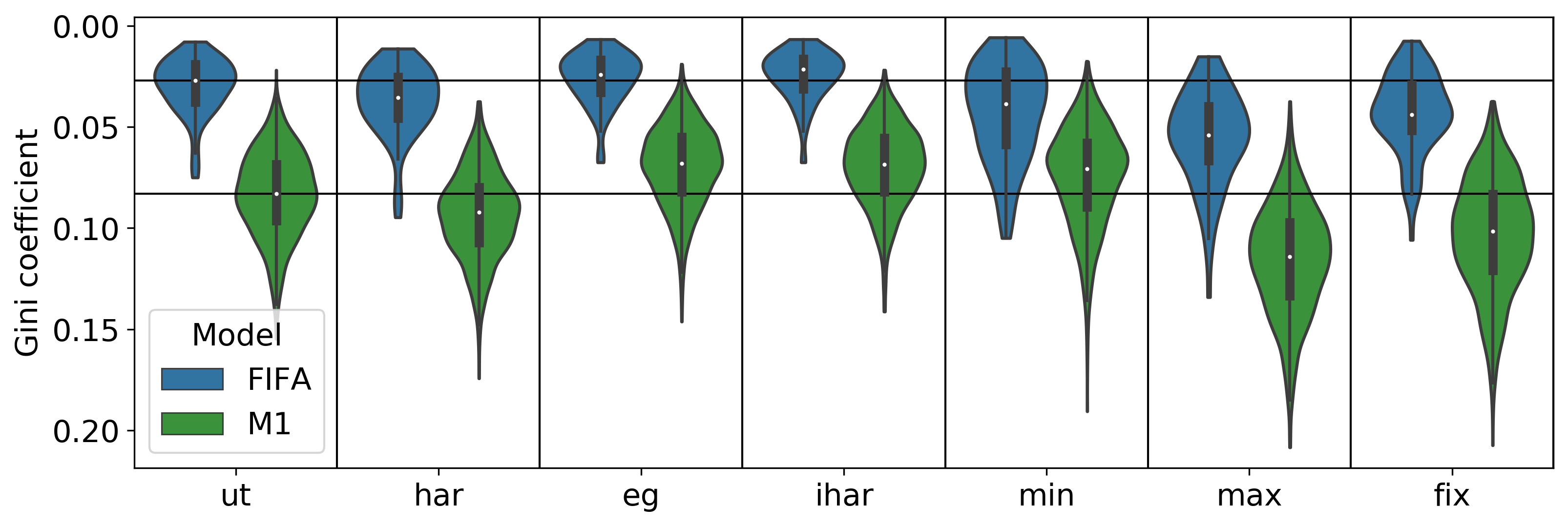}
			\caption{Gini coefficient of score vector.} 
			
		\end{subfigure}\vspace*{-0.5cm}
	\end{figure}
	\begin{figure}[htb!]\ContinuedFloat
	\begin{subfigure}[b]{\textwidth}
		\centering
		\includegraphics[width=\textwidth]{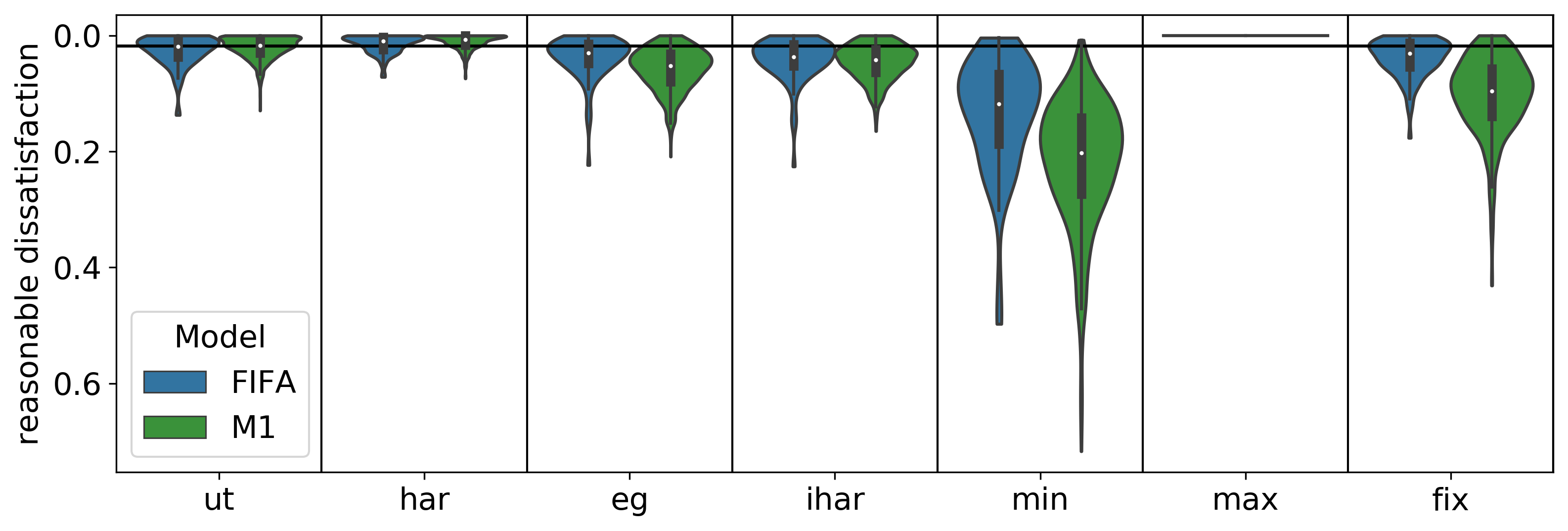}
		\caption{Reasonable dissatisfaction.} 
		
	\end{subfigure}
		\caption{Comparison of voting rules on M1 and FIFA data.}
		\label{fig:M1} \vspace*{-3cm}
	\end{figure}

	\FloatBarrier
	\clearpage
	\newpage
	\subsection{Additional Experimental Results }\label{se:moreResults}
	We also conducted experiments varying the number of positions and 
	candidates. In the following, we exemplarily include our results for the 
	case with twenty candidates on ten positions and twenty candidates on 
	twenty positions.
	\subsubsection{Twenty candidates on ten positions}
	Here, diagrams for our experiments with twenty candidates on ten positions 
	are displayed. In \autoref{fi:ins1020}, instances produced by the three 
	different models are compared to each other. Note that the general picture 
	is similar to the comparison of the models for ten-candidates ten-positions 
	elections (see \autoref{fig:ovInst}). However, here, due to the higher 
	number of available candidates, the instances exhibit less competition, 
	as overall twenty-candidates ten-positions instances have lower conflicting 
	score, Gini coefficient, and social conflict than ten-candidates 
	ten-positions elections. Moreover, in \autoref{fi:1020R}, the performance 
	of 
	the considered voting rules with respect to the summed score of a line-up, 
	the minimum score in a line-up, the Gini coefficient of the score vector, 
	and 
	the reasonable 
	dissatisfaction in a line-up are shown. Note that the general trends 
	highlighted in the main body of the paper still hold while being less 
	visible due to less competition in the data. 
	\begin{figure}[htb!]
		\includegraphics[width=\textwidth]{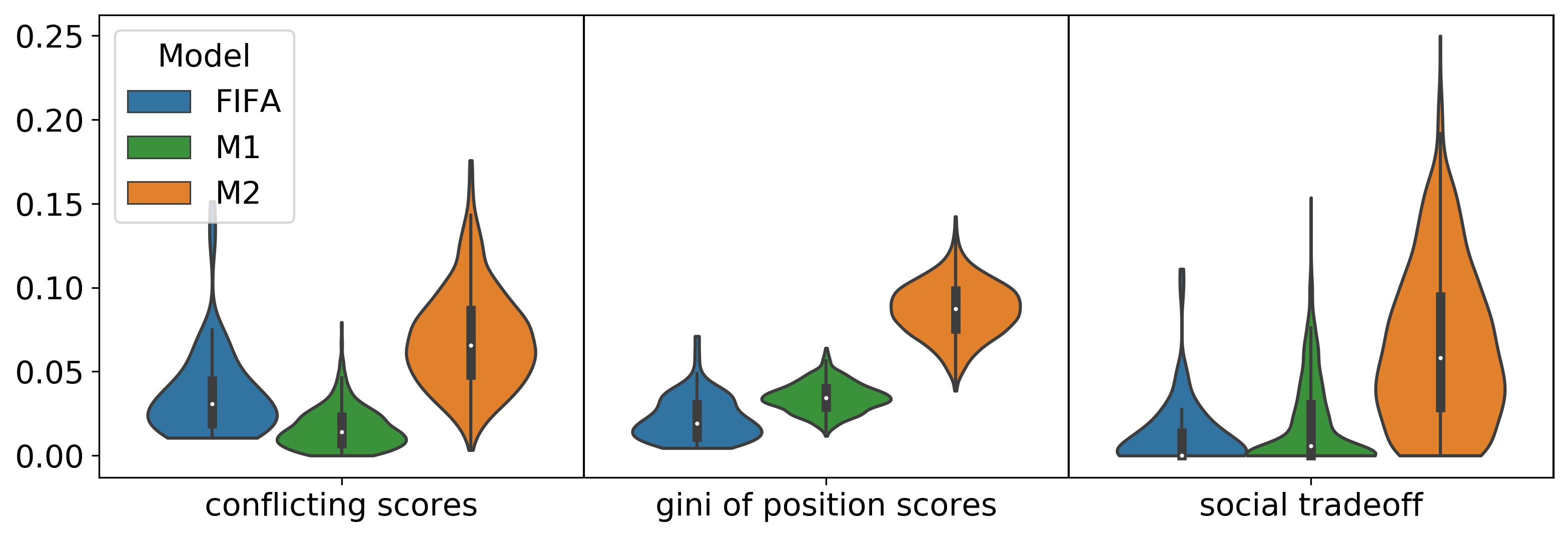}
		\caption{Comparison of different models to generate twenty-candidates 
			ten-positions elections.} \label{fi:ins1020}
	\end{figure}

	\begin{figure}[htb!]
		\captionsetup[subfigure]{justification=centering}
		\centering
		\begin{subfigure}[b]{\textwidth}
			\centering
			\includegraphics[width=\textwidth]{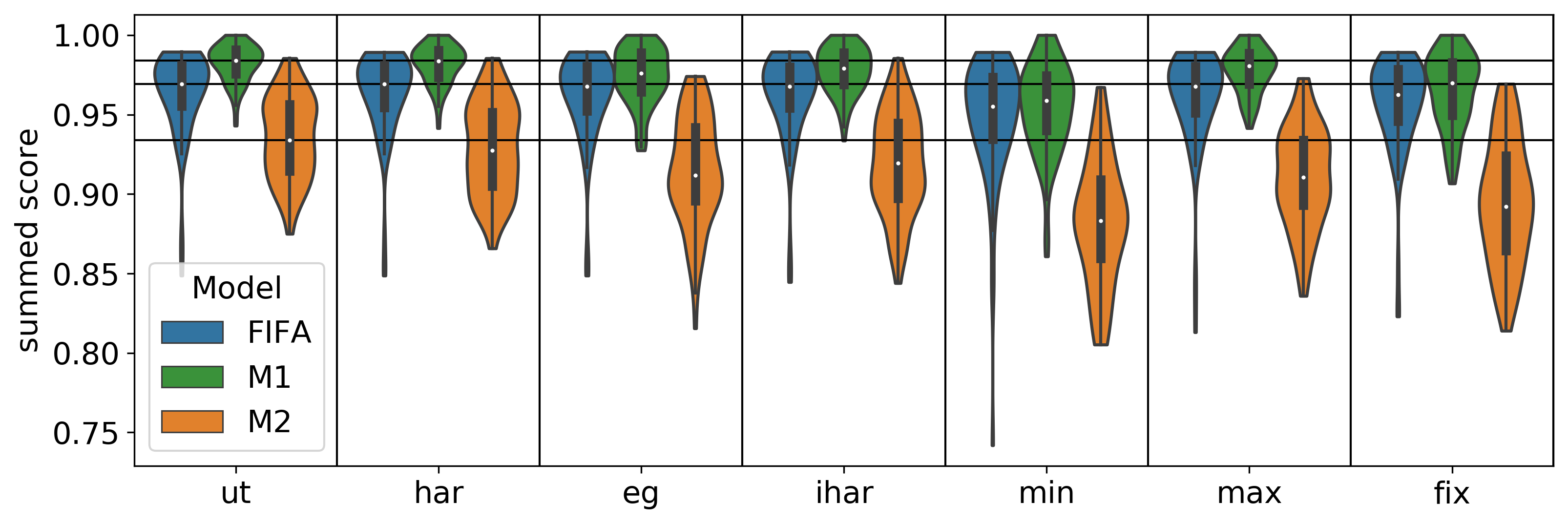}
			\caption{Summed score.} 
			
		\end{subfigure}\vspace*{-3cm}
	\end{figure}

	\begin{figure}[htb!]\ContinuedFloat
		\begin{subfigure}[b]{\textwidth}
			\centering
			\includegraphics[width=\textwidth]{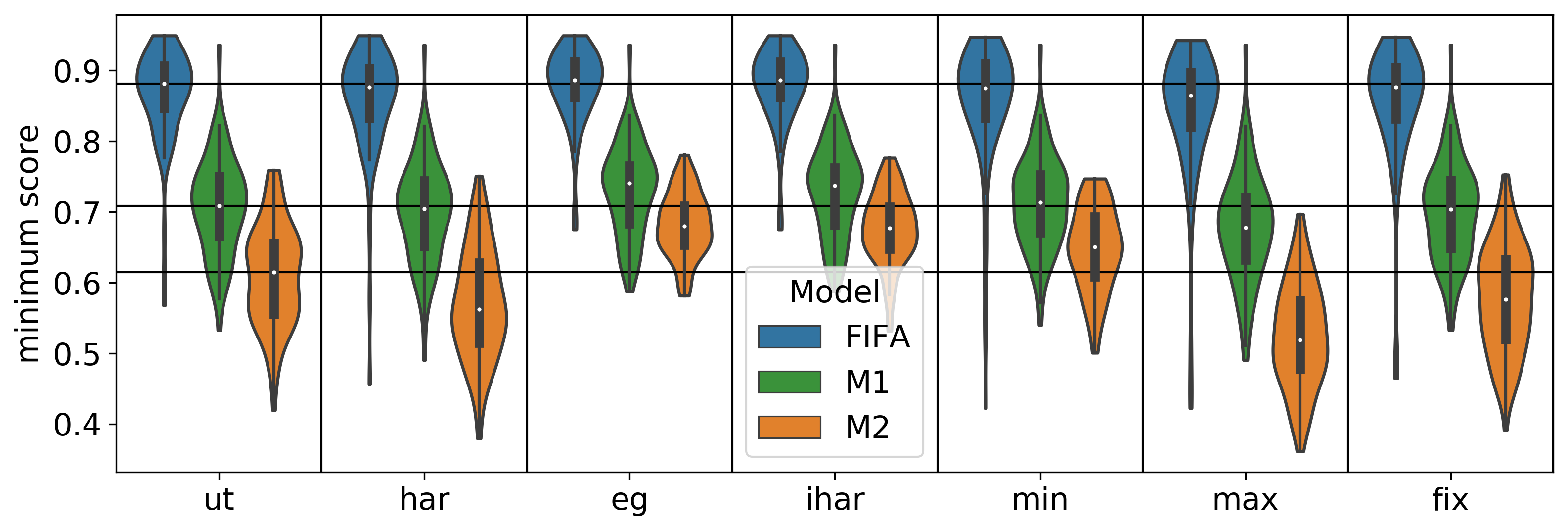}
			\caption{Minimum score.} 
			
		\end{subfigure}
	\end{figure}
	\begin{figure}[htb!]\ContinuedFloat
		\begin{subfigure}[b]{\textwidth}
			\centering
			\includegraphics[width=\textwidth]{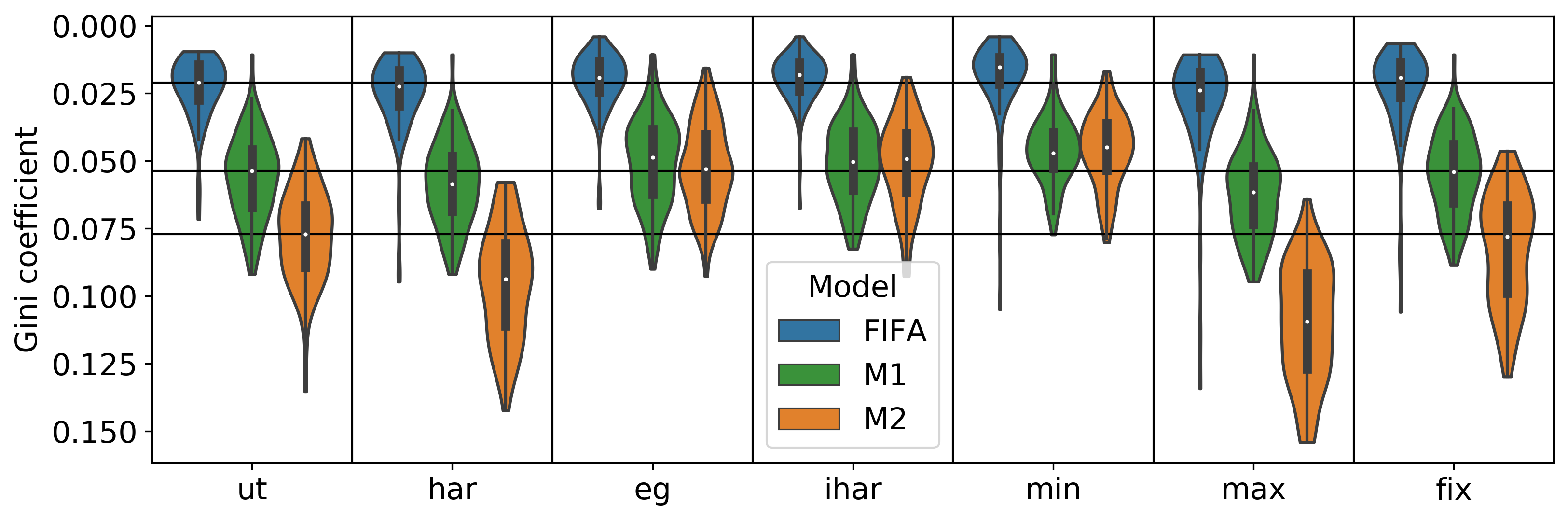}
			\caption{Gini coefficient of score vector.} 
			
		\end{subfigure}
	\end{figure}
	\begin{figure}[htb!]\ContinuedFloat
		\begin{subfigure}[b]{\textwidth}
			\centering
			\includegraphics[width=\textwidth]{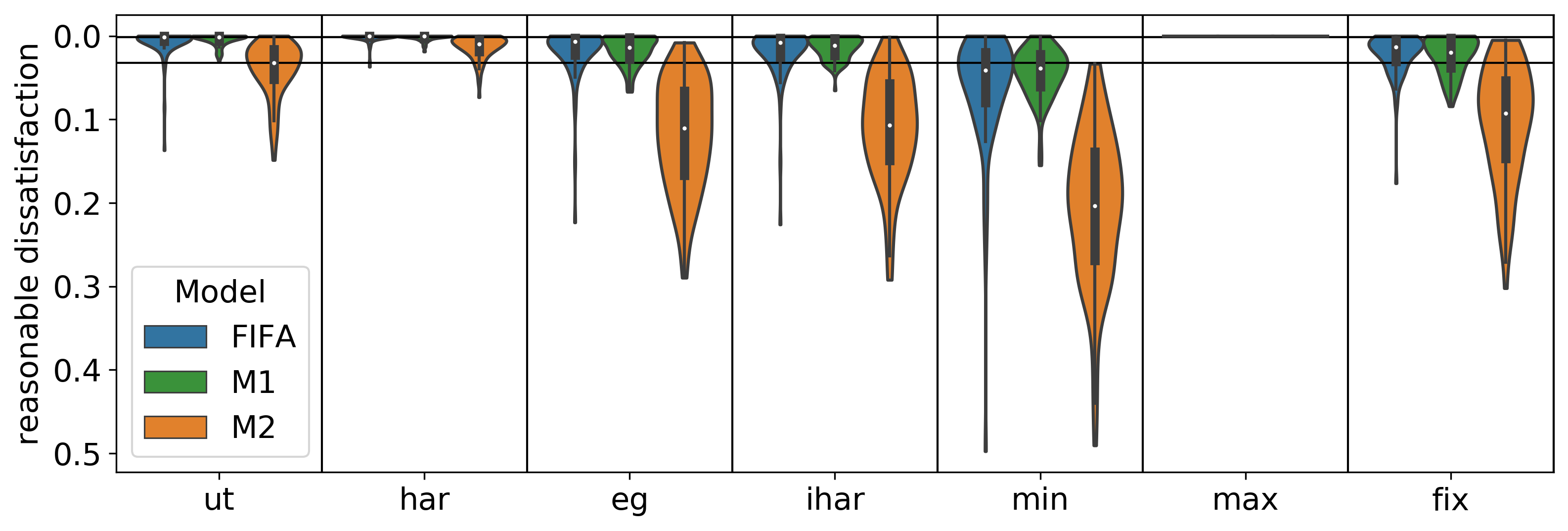}
			\caption{Reasonable dissatisfaction.} 
			
		\end{subfigure}
		\caption{Comparison of voting rules on twenty-candidates ten-positions 
		elections using M1, M2, and FIFA data.} \label{fi:1020R}
	\end{figure}
	\FloatBarrier

	\clearpage
	\subsubsection{Twenty candidates on twenty positions}
	In the following, diagrams presenting our results for experiments with 
	twenty candidates on twenty positions 
	are shown. We did not use our FIFA data for these experiments, as in soccer 
	line-ups consist only of ten (or eleven) positions. Moreover, due to 
	limited 
	computation time, we did not run the experiments in this setting for 
	the inverse harmonic rule. That is why we omit this rule in all diagrams. 
	
	In \autoref{fi:2020I}, the instances produced by the three 
	different models are compared to each other. Considering the selected three 
	metrics, the generated twenty-candidates twenty-positions instances seem to 
	be pretty similar to the generated ten-candidates ten-positions instances.  
	Additionally, the outcomes produced by  
	all voting rules except the inverse harmonic rule are analyzed in 
	\autoref{fi:1020R}. Again, the results are pretty similar to the 
	ten-candidates ten-positions case discussed in the main body of the paper. 
	\begin{figure}[htb!]
		\includegraphics[width=\textwidth]{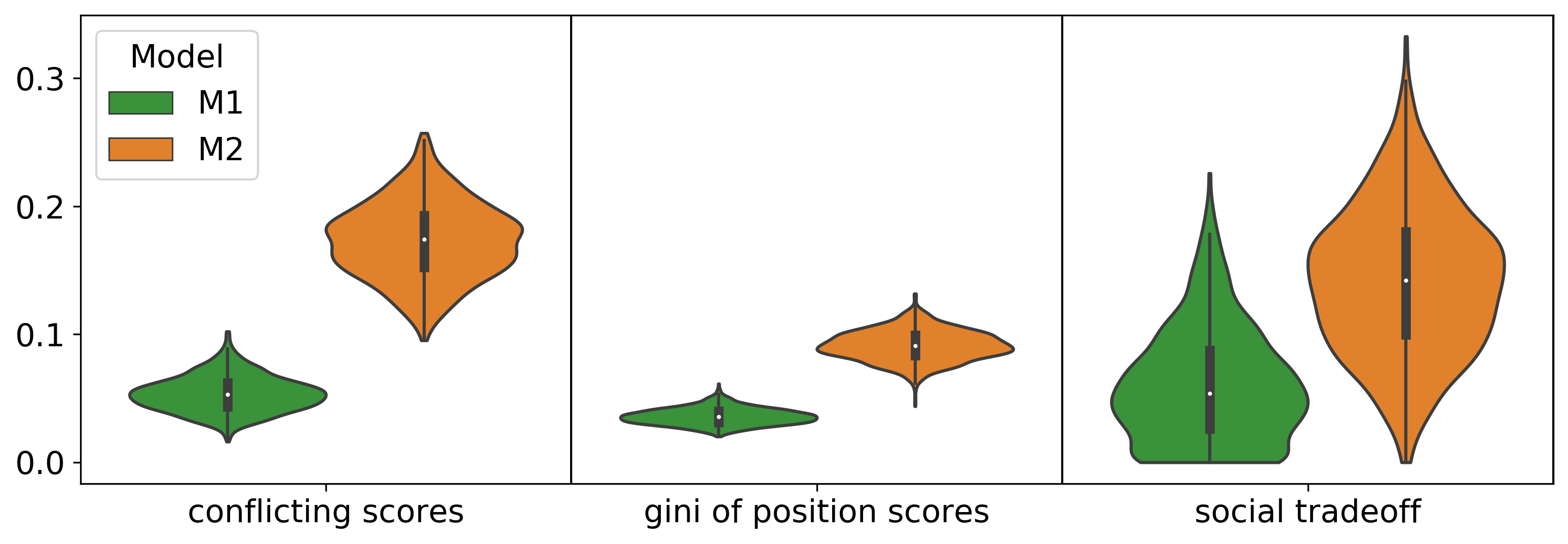}
		\caption{Comparison of different models to generate twenty-candidates 
			twenty-positions elections.} \label{fi:2020I}
	\end{figure}

	\begin{figure}[htb!]
		\captionsetup[subfigure]{justification=centering}
		\centering
		\begin{subfigure}[b]{\textwidth}
			\centering
			\includegraphics[width=\textwidth]{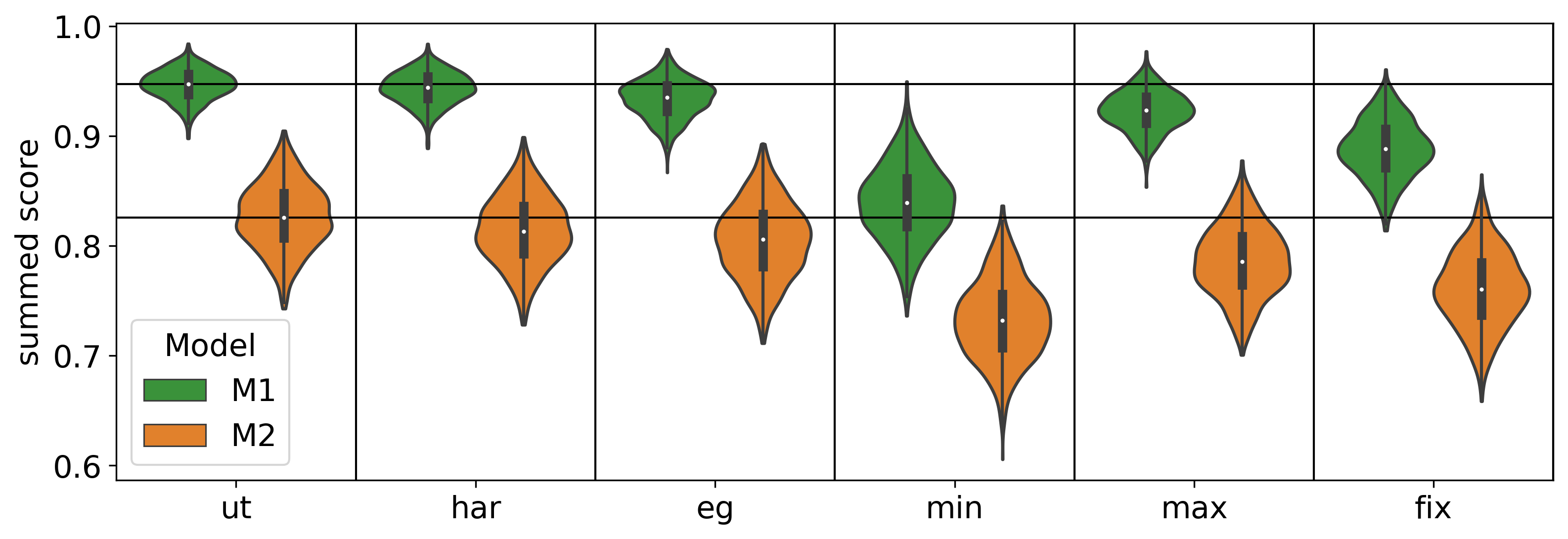}
			\caption{Summed score.} 
			
		\end{subfigure}
	\end{figure}
	\begin{figure}[htb!]\ContinuedFloat
		\begin{subfigure}[b]{\textwidth}
			\centering
			\includegraphics[width=\textwidth]{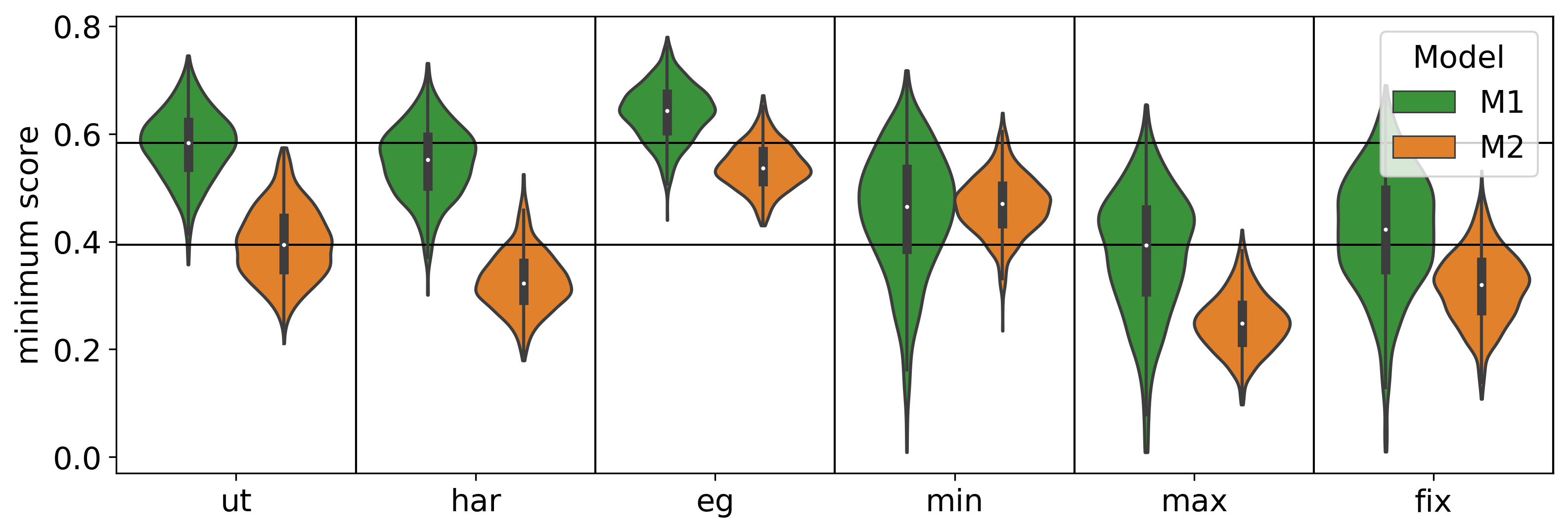}
			\caption{Minimum score.} 
			
		\end{subfigure}
	\end{figure}
	\begin{figure}[htb!]\ContinuedFloat
		\begin{subfigure}[b]{\textwidth}
			\centering
			\includegraphics[width=\textwidth]{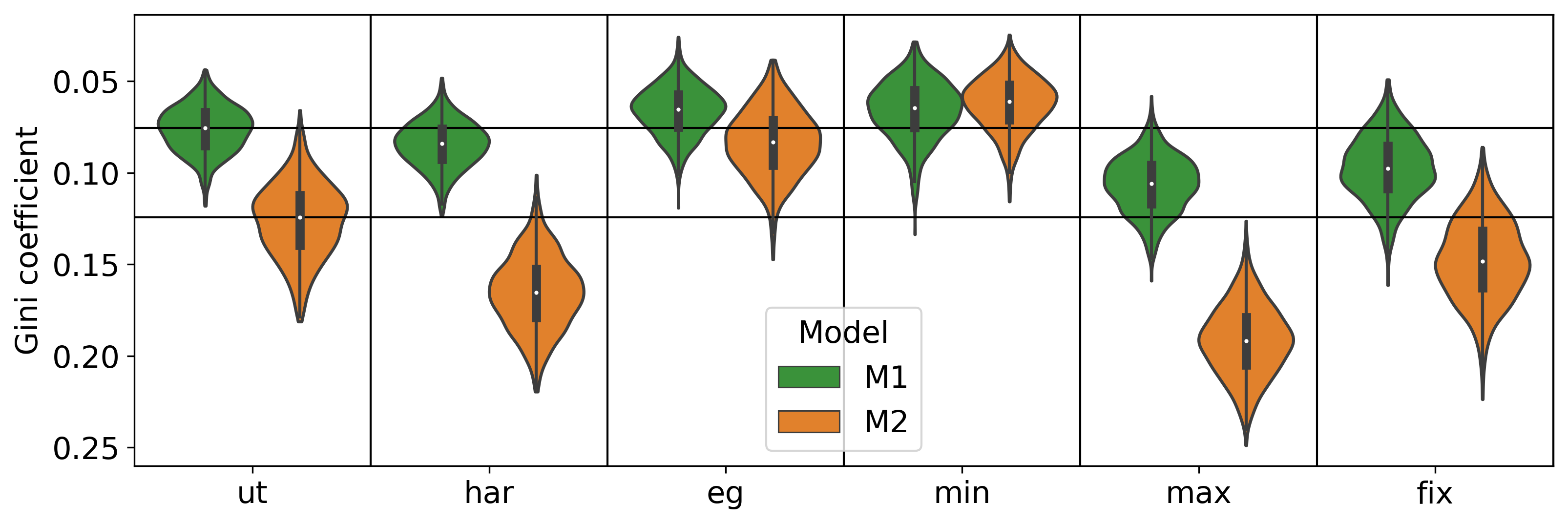}
			\caption{Gini coefficient of score vector.}
			
		\end{subfigure}
	\end{figure}
	\begin{figure}[htb!]\ContinuedFloat
		\begin{subfigure}[b]{\textwidth}
			\centering
			\includegraphics[width=\textwidth]{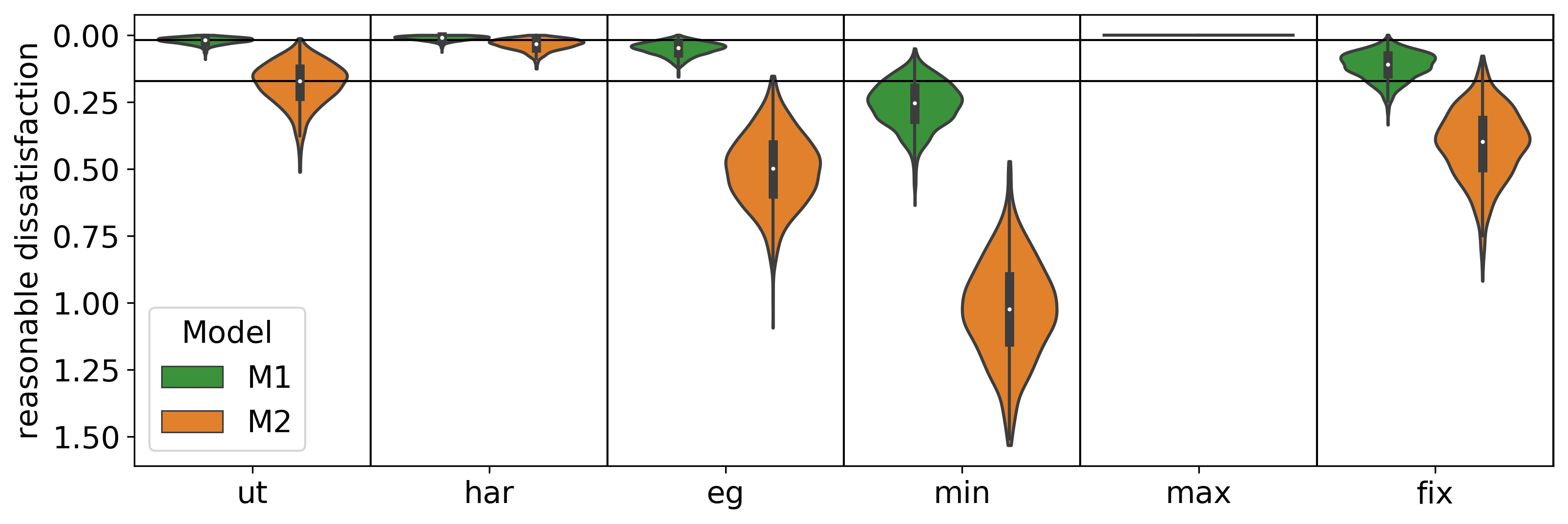}
			\caption{Reasonable dissatisfaction.} 
			
		\end{subfigure}
		\caption{Comparison of voting rules on twenty-candidates 
		twenty-positions 
			elections using M1 and M2 data.}
		\label{fig:ex1}
	\end{figure} \label{fi:2020R}
	\FloatBarrier

\end{document}